\documentclass{article}

\newcommand{\llncs}[1]{}

\usepackage{amsmath,amssymb,graphicx, amsthm}
\usepackage{enumerate}
\usepackage{url}
\newtheorem{theorem}{Theorem}[section]

\newtheorem{definition}[theorem]{Definition}
\newtheorem{lemma}[theorem]{Lemma}

\newcommand{\IEvalT}[2]{I(#1)(#2)} % interval evaluation of terms
\newcommand{\Diam}[1]{{diam(#1)}}

\newcommand{\R}{\mathbb{R}}
\newcommand{\C}{\mathbb{C}}
\newcommand{\N}{\mathbb{N}}
\newcommand{\Z}{\mathbb{Z}}

\newcommand{\sign}{\mathrm{sign\,}}
\renewcommand{\deg}{\mathrm{deg\,}}

\newcommand{\hs}{\hspace*{0.3cm}}

\newcommand{\ud}{\uparrow\downarrow}
\newcommand{\hr}{\hookrightarrow}

\begin{document}
\title{Effective Topological Degree Computation Based on Interval Arithmetic}
\author{Peter Franek and Stefan Ratschan\\
Institute of Computer Science of the\\ Academy of Sciences of the Czech Republic}

   \maketitle
\begin{abstract}
  We describe a new algorithm for calculating the topological degree
  $\deg(f,B,0)$ where $B\subseteq\R^n$ is a product of closed real intervals and
  $f:B\to\R^n$ is a real-valued continuous function given in the form of
  arithmetical expressions. The algorithm cleanly separates numerical from
  combinatorial computation. Based on this, the numerical part provably computes
  only the information that is strictly necessary for the following
  combinatorial part, and the combinatorial part may optimize its computation
  based on the numerical information computed before.  We present
  computational experiments based on an implementation of the algorithm. 
  In contrast to previous work, the algorithm does not assume knowledge of a
  Lipschitz constant of the function $f$, and works for arbitrary continuous
  functions for which some notion of interval arithmetic can be defined.
\end{abstract}

\section{Introduction}
The notion of topological degree was introduced by Jan Brouwer~\cite{Brouwer:1911} 
and was motivated by questions in differential topology~\cite{Milnor:97,Hirsch:76}.  The degree of a continuous function is an integer, describing
some topological properties of it. Degree theory has many applications,
including geometry~\cite{Sauvigny:06}, nonlinear differential equations~\cite{Krawcewicz:97,Milota:07,Mawhin:79,brown2004}, dynamical
systems~\cite{Katok:96}, verification theory~\cite{Kearfott:04}, fixed point
theory~\cite{Cronin:64} and others.

The presented algorithm is able to calculate the degree $\deg(f,B,0)$ of any
real-valued continuous function $f$ defined on a box $B$ such that $0\notin
f(\partial B)$ and $f$ is~given in~the form of arithmetical expressions containing
function symbols for which interval enclosures can be computed~\cite{Moore:09,Rump:10}.
Computational experiments show that for low-dimensional examples of simple functions (up to dimension 10) 
the algorithm terminates in reasonable time.\footnote{The program is accessible on {\tt
    topdeg.sourceforge.net}.} In addition to efficiency, the algorithm has
several advantages over previous work that we now describe in more detail.

The idea of computing the degree algorithmically is not new. 
Since the seventies, many algorithms were proposed and implemented that calculate the degree $\deg(f,B,0)$ of a function $f$ 
defined on a bounded set $B\subseteq\R^n$ via a symbolical expression. 
However, all these methods have various restrictions. 
One of the first such methods was proposed by Erdelsky in 1973~\cite{Erdelsky:73}. His~assumption is that the function
is Lipschitz, with a known Lipschitz constant.
Thomas Neil published another method for automatic degree computation in 1975~\cite{Neil:75}. It is based on the approximation of 
a multidimensional integral of a function derived from $f$ and its partial derivatives.
Here, the error analysis uses only probabilistic methods.
Other authors constructed algorithms that cover the boundary $\partial B$
with a large set of $(n-1)$-simplices 
%such that the components of $f$ are nonzero on the vertices of those simpleces, 
and use the information about the signs of
$f_j$ on the vertices of these simplices to calculate the degree in a combinatorial way. 
However, the calculated result is proved to be correct only if a parameter is chosen to be
sufficiently large~\cite{Kearfott:79,Stenger:75,Stynes:79}.
Boult and Sikorski developed a different method for degree calculation in the eighties, but their algorithm also requires the knowledge of a Lipschitz constant for $f$~\cite{Boult:86}.
Later, many algorithms arose where the degree was calculated recursively from partial information about $f$ on the boundary $\partial B$.
For example, one has 
$$\deg(f,B,0)=\deg(f_{\neg 1}, U,0),$$ where $f_{\neg 1}=(f_2,\ldots, f_n)$ and $U$ is a $(d-1)$-dimensional open neighborhood of 
$\{x\in\partial B|\, f_{\neg 1}(x)=0,\,\, f_1(x)>0\}$ in $\partial B$. 
Aberth described an algorithm using this formula, based on interval
arithmetic~\cite{Aberth:94}. This method was not implemented and is rather a recipe than a precise algorithm.
Later, Murashige published a method for calculating the degree that uses concepts from computational homology theory~\cite{Murashige:06}.

Although a broad range of ideas and methods for automatic degree computation has been implemented, 
the effectivity of these algorithms decreases fast with the dimension of $B$.
For example, in the Murashige homological method, computation of the degree of the identity function $f(x)=x$
takes more than 100 seconds already in dimension 5~\cite[Figure 3]{Murashige:06}. Other approaches were developed that calculate 
the degree of high-dimensional examples quickly, provided the functions are of some special type. 
For instance, there exist effective degree algorithms for complex functions
$f: \C^n\to\C^n$~\cite{Dian:03,Kearfott:00,Kearfott:04}. 

Our approach is based on a formalization, extension, and implementation of the
rough ideas of Oliver Aberth~\cite{Aberth:94}.  In our setting, we assume that
the function $f$ is real valued and continuous, and it is possible to implement
an interval-valued function which computes box enclosures for the range of $f$
over a box.  We don't require the function to be differentiable and not even
Lipschitz. This enables us to work with algebraic expressions containing
functions such as $\sqrt[3]{x}$, $|x|$ and $x\sin\frac{1}{x}$, but also
with any function $f$ that cannot be defined by algebraic expressions and only an algorithm
is given that computes a superset $J$ of $f(I)$ for any interval $I$ s.t. 
the measure of $J\setminus f(I)$ can be arbitrary small for small intervals $I$. 
Throughout the
paper, we assume that the domain of the function $f$ is a box (product of
compact intervals), but the algorithm works without major changes for more
general domains, such as finite unions of boxes with more complicated topology.
This will be discussed at the end of Section~\ref{sec:proof-correctness}.

From the algorithmic point of view, our algorithm consists of a numerical part,
that provably computes only information that is strictly necessary for
determining the degree, and a combinatorial part that computes the degree from
this information. The separation of those two parts has the advantage that both
can be used and improved independently. The first, numerical part 
covers the boundary of a $d$-dimensional set $\Omega$ with $(d-1)$-dimensional 
regions $D_1,\ldots, D_m$ where 
a particular component $f_l$ of $f$ has constant sign. The combinatorial part recursively gathers the information about the signs of
the remaining components of $f$ on $\partial D_j$. All the sets are represented
as lists of oriented boxes. They do not  
have to represent manifolds and we allow the boundary of these sets to be complicated (see Def. \ref{oriented_cubical_set}).
In this setting, it is computationally nontrivial to identify the boundary $\partial D_j$
of a $d$-dimensional set embedded in $\R^n$ and to decompose the boundary into a sum of ``nice'' sets. Instead of doing this, 
we calculate an ``over-approximation'' of $\partial D_j$ that is algorithmically simpler 
and then prove that it has no impact on the correctness of the result.
This involves some theoretical difficulties whose solution necessitates the
development of several technical results.

Some interest in automatic degree computation is motivated by verification
theory. Methods have been developed for automatic verification of the
satisfiability of a system of $n$ nonlinear equations in $n$ variables, written
concisely as $f(x)=0$, where $f: B\subseteq \R^n\to \R^n$ is a continuous
function.  Most of these methods first find small boxes $K$ that potentially
contain a root of~$f$ and then try to formally prove the existence of a~root
in~such a box $K$~\cite{Miranda:41,Rall:80,Frommer:2004,Frommer:05} using
tests based on theorems such as the Kantorovich theorem, Miranda theorem, or
Borsuk theorem. From those, the test based on Borsuk theorem is the most
powerful~\cite{Alefeld:01a,Frommer:05}. It can be easilly shown that
the assumptions of Miranda theorem imply that $\deg(f,K,0)=\pm 1$ and
the assumption of Borsuk theorem imply that the degree is an odd number.
It~is well known that $\deg(f,K,0)\neq 0$ implies the existence of a root of $f$ in~$K$.  
An efficient test developed by Beelitz can verify that the degree is $\pm 1$, if it is $\pm 1$, 
and hence prove the existence of a solution~\cite{Beelitz:80}. 
By not restricting oneself to degree $\pm 1$ but computing the degree in general, one can prove 
the existence of a~root of $f$ in all cases that are robust in a certain 
sense~\cite{Collins:08b,Franek:12}.

The second section contains the main definitions needed from topological degree
theory---Theorem \ref{thetheorem} is a fundamental ingredient of our algorithm. Section 3
describes the algorithm itself and its connection to Theorem \ref{thetheorem}. 
In Section 4, we present some experimental results. The last section contains the proof of two
auxiliary lemmas that we need throughout the paper. These proofs do not involve
deep ideas but are quite long and technical---hence the separate section at the
end of the paper.

\section{Mathematical Background}

\subsection{Definitions and Notation}
\label{degree_definitions}
In this section, we first summarize the definition and main characteristics of
the topological degree on which there exists a wide range of
literature, such as~\cite{Fonseca:95,Cho:06}.  Degree theory works with continuous
maps between oriented manifolds, and in order to represent these topological
objects on computers we will then introduce Definitions~\ref{boxdef}
to~\ref{def:intcomp}. Finally, the original Theorem~\ref{thetheorem} will be the
main ingredient of our algorithm for computing the topological degree.

Let $\Omega\subseteq\R^n$ be open and bounded, $f:\bar\Omega\to\R^n$ continuous
and smooth (i.e., infinitely often differentiable) in $\Omega$, $p\notin
f(\partial\Omega)$. For regular
values $p\in\R^n$ (i.e., values $p$ such that for all $y\in f^{-1}(p)$, $\det
f'(y)\not= 0$), the degree $\deg(f,\Omega, p)$ is defined to be
\begin{equation}
\label{def_deg}
\deg (f,\Omega,p):=\sum_{y\in f^{-1}(p)} \sign \det f'(y).
\end{equation}

This definition can be extended for non-regular values $p$ in a unique way, such
that for given $f$ and $\Omega$, $\deg(f,\Omega,p)$---as a function in $p$---is locally constant on the connected components of 
$\R^n\setminus f(\partial\Omega)$~\cite{Milnor:97}. 

Here we give an alternative, axiomatic definition, that determines the degree uniquely. 
For any \emph{continuous} function $f: \bar\Omega\to\R^n$ s.t. $0\notin f(\partial\Omega)$
the degree $\deg(f,\Omega,p)$ is the unique integer satisfying the following properties~\cite{Fonseca:95,Cho:06,Furi:2010}:
\begin{enumerate}
\item For the identity function $I$, $\deg(I,\Omega,p)=1$ iff $p$ is in the interior of $\Omega$.
\item\label{item:degsol} If $\deg (f,\Omega, p)\neq 0$ then $f(x)=p$ has a solution in $\Omega$.
\item\label{item:deghomot} If there is a continuous function (a ``homotopy'') $h: [0,1]\times\bar\Omega\to\R^n$ such that
 $p\notin h([0,1]\times\partial\Omega)$, then $\deg(h(0, \cdot),\Omega,p)=\deg(h(1, \cdot),\Omega,p)$.
\item If $\Omega_1\cap\Omega_2=\emptyset$ and $p\notin f(\bar\Omega\setminus(\Omega_1\cup\Omega_2))$,
then $\deg(f,\Omega,p)=\deg(f,\Omega_1,p)+\deg(f,\Omega_2,p)$.
\item For given $f$ and $\Omega$, $\deg(f,\Omega,p)$---as a function in $p$---is constant on any connected component of $\R^n\backslash f(\partial\Omega)$.
\end{enumerate}

This can be generalized to the case of a continuous function $f:M\to N$, where $M$ and $N$ are oriented manifolds of the same dimension
and $M$ is compact. If $f$ is smooth, $f'(y)$ denotes the matrix of partial 
derivatives of some coordinate representation of $f$ and formula $(\ref{def_deg})$ is still meaningful.   
For example, if $f$ is a scalar valued function from
an oriented curve $c$ (i.e., an oriented set of dimension $1$) to $\R$ and
$f\neq 0$ on the endpoints of $c$, then $\deg(f,c,0)$ is well defined. 
If $f: M\to N$ is a function between two oriented manifolds without boundary, then the degree $\deg(f)$ is defined to be $\deg(f,M,p)$ for any $p\in f(M)$.

A simple consequence of the degree axioms is that for a continuous $f: \bar\Omega\subseteq\R^n \to \R^n$, 
$p\notin f(\partial\Omega)$ implies that $\deg(f,\Omega,p)=\deg(f-p,\Omega,0)$. 
So we will be only interested in calculating $\deg(f,\Omega,0)$.

We will represent geometric objects like manifolds, orientation, boundaries and functions in a combinatorial way, using
the following definitions.

\begin{definition}
\label{boxdef}
A $k$-dimensional box (simply $k$-box) in $\R^n$ is the product of $k$
non-degenerate closed intervals and $n-k$ degenerate intervals (one-point
sets). A \emph{sub-box} of a $k$-box $A$ is any $k$-box $B$ s.t. $B\subseteq
A$. The \emph{diameter} $\Diam{B}$ of a box $B$ is the width of its widest interval.
\end{definition}

\begin{definition}
\label{def:oriented_box}
The \emph{orientation} of a $k$-box is a number from the set $\{1, -1\}$. 
An \emph{oriented box} is a pair $(B,s)$ where $B$ is a box and $s$ its orientation.
We say that $B_1$ is an \emph{oriented sub-box} of an oriented box $B$, if $B_1\subseteq B$, 
the dimensions of $B$ and $B_1$ are equal and the orientations are equal.
\end{definition}

\begin{definition}
\label{boundary_orientation} 
Let $B=I_1\times I_2\times \ldots \times I_n$ be an oriented $d$-box in $\R^n$
with orientation $o$. Let, for every $i\in \{1,\dots,n\}$, $[a_i,b_i]=I_i$.
Assume that the intervals $I_{j_1},\ldots I_{j_d}$ are non-degenerate,
$j_1<j_2<\ldots <j_d$, the other
intervals are degenerate (one-point) intervals. 
For $i\in\{1,\dots, d\}$, 
the $(d-1)$-dimensional boxes
$$F_i^-:=\{(x_1,\dots,x_n)\in B\,|\, x_{j_i}=a_{j_i}\}\quad {\rm and}\quad
F_i^+:=\{(x_1,\dots,x_n)\in B\,|\, x_{j_i}=b_{j_i}\}$$ are called \emph{faces} of $B$.
Any sub-box of a face is called a \emph{sub-face} of $B$. 
If we choose the orientation of 
$F_i^+$ to be $(-1)^{i+1} o$ and the orientation of $F_i^-$ to be $(-1)^{i} o$, then we call $F_i^\pm$
\emph{oriented faces} of $B$. An oriented sub-box of an oriented face is called
\emph{oriented sub-face}. The orientation of the oriented faces and sub-faces is called the \emph{induced orientation} from the orientation of $B$.
\end{definition}

\begin{definition}
\label{oriented_cubical_set}
An \emph{oriented cubical set} $\Omega$ is a finite set of oriented boxes
$B_1,\ldots, B_k$ of the same dimension $d$ such that the following conditions
are satisfied:
\begin{enumerate}
\item For each $i\neq j$, the dimension of  $B_i\cap B_j$ is at most $(d-1)$.
\item Whenever $B_i\cap B_j=B_{ij}$ is
a $(d-1)$-dimensional box, then the orientations of $B_i$ and $B_j$ are compatible. 
This means that $B_{ij}$ has an opposite induced orientation as a sub-face of $B_i$ as the orientation
induced from $B_j$.
\end{enumerate}
The dimension of an oriented cubical set is the dimension of any box it contains. If $\Omega$ is an oriented cubical set, we denote by $|\Omega|$ 
the set it represents (the union of all the oriented boxes contained in $\Omega$).
\end{definition}

An oriented cubical set is sketched in Figure~\ref{fig:oc}. An immediate
consequence of the definition is that each sub-face 
$F$ of a box $B$ in an oriented cubical set $\Omega$ is a boundary sub-face of at most two boxes in $\Omega$.
\begin{figure}[htbp]
\centering
\includegraphics[width=6cm]{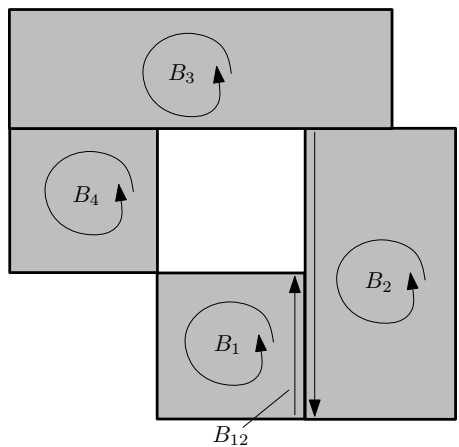}
\caption{Two-dimensional oriented cubical set, union of four oriented
     boxes. The boundary face $B_{12}$, for example, has opposite orientation induced from the box $B_1$ and from $B_2$.}
\label{fig:oc}
\end{figure}
Note that an oriented cubical set does not have to represent a manifold, because some
boxes may have lower-dimensional intersection, like $B_1$ and $B_4$ in Figure
$\ref{fig:oc}$. %\footnote{will the interior of $\Omega$ always be a manifold?}
%
%Now we introduce the combinatorial version of the geometrical definition of oriented manifolds and its oriented boundary with the induced orientation. 
%
\begin{definition}
\label{def:oriented_boundary}
An \emph{oriented boundary} of an oriented $d$-dimensional cubical set $\Omega$
is any set of $(d-1)$-dimensional oriented boxes $\partial\Omega$, such that
\begin{enumerate}
\item\label{cond:1} Any two boxes in $\partial\Omega$ have intersection of dimension at most $d-2$.
\item\label{cond:2} For each $F_\partial\in\partial\Omega$, and each
  $(d-1)$-dimensional sub-box $F'$ of $F_\partial$, there exists \emph{exactly one} box $B\in\Omega$ 
such that $F'$ is an oriented sub-face of $B$. %The orientation of $F_\partial$ is induced from $B$.
%and a unique boundary face $F$ of $B$ such that $F'$ is an oriented subbox of $F$.
\item $\partial\Omega$ is maximal, that is, no further box can be added to $\partial\Omega$ such that conditions \ref{cond:1} and \ref{cond:2} still hold.
\end{enumerate}
\end{definition}
An oriented cubical set and its oriented boundary are sketched in Figure~\ref{fig:ob}. Geometrically, this definition describes the topological boundary of an oriented cubical set $\Omega$ and
we denote the union of all oriented boxes in $\partial\Omega$ by $|\partial\Omega|$. 
Clearly, $\partial |\Omega|=|\partial\Omega|$, the meaning of the left hand side
being the topological boundary of the set $|\Omega|$. Note that if 
$\Omega$ is a $d$-dimensional oriented cubical set and $\partial\Omega$
an oriented boundary of $\Omega$,
then each sub-face $x$
of some box in $\Omega$ s.t. $x\cap |\partial\Omega|$ is at most 
$(d-2)$-dimensional, is a sub-face of exactly two boxes in $\Omega$ with opposite induced orientation (see $B_{12}$
in Figure~\ref{fig:oc}).

An oriented boundary of an oriented cubical set does not have to form an oriented cubical set, because 
the second condition of Definition $\ref{oriented_cubical_set}$ may be violated
(for a counter-example, see Fig. \ref{fig:lemma1} where the $1$-boxes $a$ and $c$ have 
$0$-dimensional intersection but not compatible orientations).

\begin{figure}[htbp]
  \centering
   \includegraphics[width=6cm]{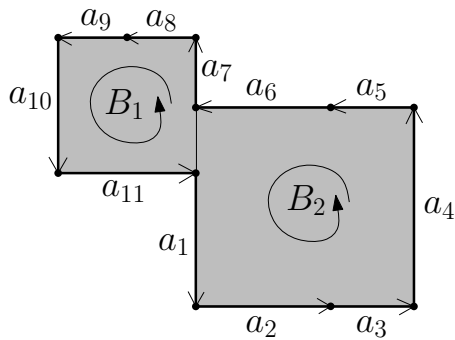}
   \caption{Oriented cubical set $\Omega=\{B_1, B_2\}$ with orientation of both boxes 1 and its oriented boundary $\partial\Omega=\{a_1,\ldots, a_{11}\}$ with orientation
indicated by the arrows. For example, $a_1$ has orientation $-1$ (the arrow goes in the opposite direction as the vertical axis), $a_2$ has orientation $1$ etc.}
   \label{fig:ob}
\end{figure}

The notion of topological degree can be naturally generalized to oriented cubical sets. So, if $f$ is a continuous function from a $d$-dimensional oriented cubical set $\Omega$ 
to $\R^d$ such that $0\notin f(\partial |\Omega|)$, then $\deg(f,\Omega,0)$ is well-defined, extending the definition of $\deg(f,\Omega,0)$ for oriented manifolds $\Omega$.
\footnote{For an oriented cubical set $\Omega$, one can define an oriented manifold $\Omega^\epsilon:=\{x\in int(|\Omega|)\,|\,dist(x,\partial\Omega)\geq \epsilon\}$ for a small enough $\epsilon$
and define the degree to be $\deg(f,\Omega^\epsilon,0)$}

Finally, we will represent functions as algorithms that can calculate 
a superset of $f(B)$ for any given box $B$. 
\begin{definition}
\label{def:intcomp}
Let $\Omega\subseteq\R^n$. We call a function $f:\Omega\rightarrow\R$ \emph{interval-computable} if there
exists a corresponding algorithm $I(f)$ that, for a given box $B\subseteq\Omega$ with rational endpoints and positive diameter,
computes a closed (possibly degenerate) interval $\IEvalT{f}{B}$ such that 
\begin{itemize}
\item $I(f)(B)\supseteq \{ f(x) \mid x\in B\}$, and 
\item for every $\varepsilon>0$ there is a $\delta>0$ such that for every box $B$
  with $0<\Diam{B}<\delta$, $\Diam{\IEvalT{f}{B}}<\varepsilon$.
\end{itemize}
We call a function  $f=(f_1,\ldots,f_n):\Omega\rightarrow\R^n$ interval-computable iff each $f_i$ is interval-computable. 
In this case, the algorithm $I(f)$ returns a
tuple of intervals, one for each $f_i$.
\end{definition}

Usually such functions are written in terms of symbolic expressions containing
symbols denoting certain basic functions such as rational constants, addition,
multiplication, exponentiation, trigonometric function and square root. Then, $I(f)$ can be computed from the
expression by interval arithmetic~\cite{Neumaier:90,Moore:09}. The interval literature usually
calls an interval function fulfilling the first property of Definition~\ref{def:intcomp}
``enclosure''. Instead of the second property, it often uses a slightly stronger notion of an
interval function being ``Lipschitz continuous''~\cite[Section 2.1]{Neumaier:90}.
We will use interval computable functions and expressions
denoting them interchangeably and assume that for an expression denoting a
function $f$, a corresponding algorithm $I(f)$ is given.

\subsection{Main Theorem}
%\begin{definition}
%Let $f:A\to \R^n$
%\end{definition}

Now we define the combinatorial information we use to compute the degree, and prove
that it is both necessary and sufficient for determining the degree.

\begin{definition}
A~$d$-dimensional \emph{sign vector} is a vector from $\{ -,0, +\}^d$. 

%Let $B$ be a $d$-dimensional oriented cubical set and $\partial B$ its oriented boundary.
Let $S$ be a set of oriented $(d-1)$-boxes.  A \emph{sign covering} of $S$ is an assignment of a $d$-dimensional sign vector
to each $a\in S$. For a sign covering $SV$ and $a\in S$ we will denote
this sign vector by $SV_a$, and its $i$-th component by $(SV_a)_i$.

A sign covering is \emph{sufficient} if each sign vector contains at least one
non-zero element.

A sign covering is a \emph{sign covering wrt. a function}  $f: (\cup_{a\in S}\, a) \to\R^{d}$ 
with components $(f_1,\ldots, f_{d})$,  if for every oriented box $a\in S$
and for every $i\in\{1,\dots,d\}$, $(SV_a)_i\neq 0$ implies that $f_i$ has constant sign $(SV_a)_i$ on $a$.
\end{definition}

In the following we will often recursively reduce proofs/algorithms for $d$-dimension
oriented cubical sets, to proofs/algorithms on their oriented
boundary. Since---as we have already seen---an oriented boundary 
of an oriented cubical set does not necessarily have to form an oriented cubical
set, we will need the following lemma that will allow us to decompose this oriented
boundary again into oriented cubical sets:

\begin{lemma}
\label{lemma:1}
Let $\Omega$ be a~$d$-dimensional oriented cubical set, $\partial\Omega$ an oriented
boundary of $\Omega$, $SV$ a sufficient sign-covering of $\partial \Omega$
with respect to $f:|\Omega|\to\R^d$ and assume that for each $a\in\partial \Omega$, $SV_a$ has exactly one nonzero component. 
Let $\Lambda_{l',s'}:=\{a\in\partial \Omega\,|\,(SV_a)_{l'}=s'\}$ for each
$l'\in\{1,\ldots,d\}$ and $s'\in\{+,-\}$. 
%Let $l\in\{1,\ldots, d\}$ and $s\in\{+,-\}$.
Then there exist oriented cubical sets $D_1,\ldots, D_m$ and corresponding
oriented boundaries $\partial D_1,\dots,\partial D_m$ s.t. the following conditions are satisfied:
\begin{enumerate}
\item\label{lm:1} $\cup_{j\in\{1,\dots,m\}} D_j=\partial\Omega$,
\item\label{lm:2} $D_i\cap D_j=\emptyset$ for $i\neq j$,
\item\label{lm:3} For each $i$, there exists $l(i), s(i)$ such that $D_i\subseteq\Lambda_{l(i),s(i)}$,
\item\label{lm:4} Each $b\in \partial D_i$ is a sub-face of some
  $a\in\Lambda_{l',s'}$ where $l'\neq l(i)$.
\end{enumerate}
%\footnote{Is this now equivalent to items 3 and 4? 
%for every $i\in \{1,\dots,m\}$ there is a coordinate $l\in\{1,\ldots,d\}$
%s.t. the $l$-th entry $s$ of the sign vector $SV$ is constant and non-zero on
%$D_i$, and  each $b\in\partial D_i$ is a subface of an $a\in \partial B$, for
%which there is an $l'\neq l$ s.t. $(SV_a)_{l'}$ is
%non-zero.}
\end{lemma}
The lemma is illustrated in Figure \ref{fig:lemma1}. The proof of this lemma is technical and we postpone it to the appendix in order to keep the text fluent.
\begin{figure}[htbp]
  \centering
   \includegraphics[width=6cm]{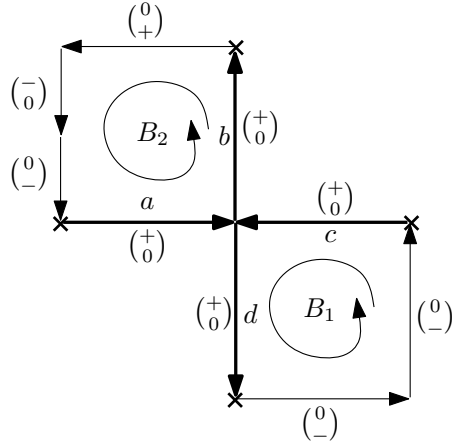}
   \caption{Illustration of Lemma \ref{lemma:1}. The boundary of the oriented cubical set $\{B_1,B_2\}$ contains nine boxes and
$\Lambda_{1,+}=\{a,b,c,d\}$ are the boundary boxes with sign vector $+\choose 0$.
This can be decomposed into oriented cubical sets $D_1=\{a,b\}$ and $D_2=\{c,d\}$. 
The boundaries of $D_1$ and $D_2$ consist of the points marked as $\times$ and each of them is a sub-face of some box
with sign vector different from $+\choose 0$.}
   \label{fig:lemma1}
\end{figure}
\begin{theorem}
\label{thetheorem}
Let $\Omega$ be an oriented $d$-dimensional cubical set, $\partial\Omega$ an oriented boundary of $\Omega$ and  
$f: |\Omega|\to \R^d$ a continuous function with components
$(f_1,\ldots, f_{d})$ such that $0\notin f(|\partial \Omega|)$.
Then a sign covering $SV$ of $\partial\Omega$ wrt. $f$ determines the degree
$\deg(f, \Omega, 0)$ uniquely if and only if it is sufficient.
\end{theorem}
\begin{proof}
We first prove that sufficiency of the sign covering implies a unique degree.
We proceed by induction on the dimension of $\Omega$.
If $\Omega$ is a~$1$-dimensional oriented cubical set $\overrightarrow{ab}$, then $\deg(f,\Omega,0)=\frac{1}{2}\,(\sign(f(b))-\sign(f(a)))$
is determined by the sufficient sign covering of $\partial\Omega$ wrt. $f$. 
Let $d>1$.
For each box $a\in\partial\Omega$, choose an index $i(a)$ such that $(SV_a)_{i(a)}=:s_a\neq 0$. 
For all $l'\in\{1,\ldots,d\}$ and $s'\in\{+,-\}$, let $\Lambda_{l',s'}:=\{a\in \partial\Omega\,|\,i(a)=l',\,s_a=s'\}$.
It follows from Lemma~\ref{lemma:1} that 
we may decompose $\partial\Omega$ into oriented cubical sets $D_j$ and oriented boundaries $\partial D_j$, $j=1,\ldots, m$ 
such that $D_i\subseteq\Lambda_{l(i),s(i)}$ for unique $l(i), s(i)$ and each $x\in\partial D_i$ is a sub-face
of some $b\in\Lambda_{l',s'}$ where $l'\neq l(i)$.
For each $l'$, define $f_{\neg l'}:=(f_1,\ldots, f_{l'-1}, f_{l'+1},\ldots, f_n)$.
Then $0\notin f_{\neg l(i)}(|\partial D_i|)$ and the degree $\deg(f_{\neg l(i)},D_i,0)$ is defined. 
Let $l\in\{1,\ldots, d\}$ and $s\in\{+,-\}$ be arbitrary.
It follows from \cite[Theorem 2.2]{Kearfott:00} and \cite[Section 4.2]{Stenger:75} that
\begin{equation}
\label{degreeformula}
\deg (f, \Omega, 0_d)=
s\,(-1)^{l+1}\sum_{i; \,l(i)=l\,\,\text{and}\,\,s(i)=s} \deg(f_{\neg l}, D_i, 0_{d-1})
\end{equation}
where $0_k\in\R^k$ is the $k$-dimensional zero.

For each set $D_i$ from the sum on the right hand side, $f_{l(i)}$ has sign $s(i)$ on $D_i$.
Each $x\in\partial D_i$ is a sub-box of some $b\in\Lambda_{l',s'}$ where $l'\neq l(i)$, 
so we may assign a new sign vector for $x$ %and $f_{\neg l(i)}$
by deleting the $l(i)$-th component from $SV_b$. In this way, we define a sufficient sign covering of $\partial D_i$ wrt. $f_{\neg l(i)}$
and the degree $\deg(f,B,0)$ can be then calculated
recursively using $(\ref{degreeformula})$.

Now assume that the sign covering of $\partial\Omega$ is not sufficient. We will prove that in this case, the degree is not uniquely determined.

Let $F\in \partial\Omega$ be a $(d-1)$-dimensional box such that $SV_F=(0,\ldots,0)$. Choose $m\in\Z$ to be arbitrary. 
We will construct a function $G: |\Omega|\to\R^d$
such that the sign covering of $\partial\Omega$ is a sign covering with respect to $G$ and $\deg(G,\Omega,0)=m$.

Denote the oriented manifold with boundary $\partial\Omega\setminus F^\circ$ by $S_1$. $\partial\Omega$ 
is a union of the oriented manifolds $S_1$ and $F$,
the boundaries $\partial F$ and $\partial S_1$ are equal with opposite orientations, homeomorphic to the sphere $S^{d-2}$. 
The degree $\deg(f,\Omega,0)=\deg(\tilde{f})$ where $\tilde{f}=f/|f|: \partial \Omega\to S^{d-1}\subseteq \R^d$ is a map to the sphere.
Let $p\in S^{d-1}$ be such that $p\notin \tilde{f}(\partial S_1)$, let $\alpha=\deg(\tilde{f},S_1,p)$ and $m'=m-\alpha$. 
We construct a map $g: F\to S^{d-1}$ such that $\deg(g,F,p)=m'$. 
The homotopy group $\pi_k(S^l)=0$ for $k<l$, so each map
from a $(d-2)$-sphere to the $(d-1)$-sphere is homotopic to a constant map.
Let us define $g_1=\tilde{f}$ on $\partial F\simeq S^{d-2}$.  Then $g_1:\partial F\to S^{d-1}$ 
is homotopic to a constant map.
There exists a sub-box $F'\subseteq F$ and a continuous extension $g_2:F\setminus (F')^\circ\to S^{d-1}$ of $g_1$ such that
$g_2=g_1=\tilde{f}$ on $\partial F$ and $g_2$ is constant on $\partial F'\simeq S^{d-2}$.
Using the fact that $\pi_{d-1}(S^{d-1})=\Z$, there exists a map $h:S^{d-1}\to S^{d-1}$ of degree $m'$.
It follows from the identity $S^{d-1}\simeq F'/\partial F'$ that we can extend $g_2$ to a map $g_3: F\to S^{d-1}$
such that $\deg(g_3,F,p)=m'$.  Finally, extend $g_3$ to a map $g:\partial \Omega\to S^{d-1}$ by $g=\tilde{f}$ on $S_1$. 
Then 
$$
\deg(g)=\deg(g,S_1,p)+\deg(g,F,p)=\alpha+m'=m.
$$
Let $i: S^{d-1}\hookrightarrow \R^d$ be the inclusion.
Multiplying $i\circ g$ by some scalar valued function, we can obtain a function $g':\partial \Omega\to \R^d$ such that $g'=f$ on $\partial\Omega$. Extending
$g': \partial \Omega\to \R^d$ to a continuous $G:\Omega\to\R^d$ arbitrarily 
(this is possible due to Tietze's Extension Theorem~\cite[Thm. 4.22]{Bruckner:1997},\cite{Tietze:1915})
we obtain a function $G$ such that the original sign covering is a sign covering of 
$\partial\Omega$ wrt. $G$ and $\deg(G,\Omega,0)=m$. This completes the proof.
\end{proof}

\section{Algorithm description}
\subsection{Informal Description of the Algorithm}
We describe now our algorithm for degree computation of an interval computable
function. If $f: B\to\R^n$ is an interval computable function nowhere zero on the boundary $\partial B$,
then the corresponding interval computation algorithm $I(f)$ from Definition~\ref{def:intcomp}
may be used to construct a sufficient sign covering of $\partial B$ wrt. $f$. 
This sign covering will be represented as a \emph{list} of oriented boxes and sign vectors. 
The main ingredient of the algorithm is Equation
$(\ref{degreeformula})$ from the proof of Theorem~\ref{thetheorem}. 
For some index $l$ and sign $s$, we select all the boxes 
$a$ with $(SV_a)_l=s$. From Lemma~\ref{lemma:1}, we know that these boxes form
some oriented cubical sets $D_1,\ldots, D_m$. 
Then a new list of $(n-2)$-dimensional oriented boxes is constructed that covers the boundaries $\partial D_j$ of $D_j$.
Possibly subdividing boxes in this new list, we assign $(n-1)$-dimensional sign vectors to its elements in such a way that 
we obtain a sufficient sign covering of $\cup_j \partial D_j$ wrt. $f_{\neg l}:=(f_1,\ldots,f_{l-1}, f_{l+1},\ldots, f_n)$. 
Equation $(\ref{degreeformula})$ is used for a recursive dimension reduction.

We work with lists of oriented boxes and sign vectors rather than with sets, 
because it will be convenient for our implementation to allow an oriented box to
be contained in a~list multiple times. However, we will usually ignore the
order of the list elements (i.e., the algorithm actually is based on
multi-sets which we implement by lists).
For two lists $L_1$ and $L_2$, we denote by $L_1+L_2$ the concatenation of $L_1$
and $L_2$ and will also use the symbol $\sum$ for the concatenation of several lists.
We use the notation $a\in L$ if $a$ is contained in $L$ at least once. If $L_1$ is a sub-list of $L$, 
we denote by $L-L_1$ 
the list $L$ with the sub-list $L_1$ omitted. %The length $|L\setminus L_1|=|L|-|L_1|$.

Now we define a version of the notion of sign covering based on lists:
\begin{definition}
\label{def:signcover}
A \emph{sign list} (of \emph{dimension} $d$) is a list of pairs consisting of
\begin{itemize}
\item an oriented $d$-box, and
\item a corresponding $(d+1)$-dimensional sign vector.
\end{itemize}

A sign list is \emph{sufficient}, iff each sign vector contains at least one non-zero element.
A sign list of dimension $d$ is a 
\emph{sign list wrt. a function} $f: \bigcup_{a\in L}\, a\,\rightarrow\mathbb{R}^{d+1}$
iff for each element $a\in L$ and corresponding sign vector
$SV_a=(s_1,\dots,s_{d+1})$, for all $i\in \{1,\dots,d+1\}$, $s_i\neq 0$ implies that
$f_i$ has sign $s_i$ on $a$.
\end{definition}

By misuse of notation, we will sometimes refer to the elements of a sign
list as pairs consisting of an oriented box and a sign vector, and sometimes just as an oriented box. 

The basic ingredient of the algorithm is a recursive function $\rm{Deg}$ with input
a sufficient sign list and output an integer. This function involves no interval arithmetic 
and is purely combinatorial. For an input that is a sufficient sign list $L$ wrt. $f$
such that the boxes in $L$ form an oriented boundary of an oriented cubical set $\Omega$, 
this function returns $deg(f,\Omega,0)$.
If the $\rm{Deg}$ function input is a $0$-dimensional sign list $L$, then the output 
$\sum_{p\in L} \frac{orientation(p)\times s_p}{2}$ is returned. This is
compatible with the the formula for the degree of a function on an oriented edge,
$\deg(f,\overrightarrow{ab},0)=\frac{\sign f(b)-\sign f(a)}{2}$.

If the input consists of oriented $d$-boxes and sign vectors of dimension $d+1$
for $d>0$, we choose $l\in \{1,2,\ldots, d+1\}$ and $s\in \{+,-\}$ and compute a
list of boxes $L^{sel}$ (the \emph{selected boxes}) having $s$ as the $l$-th component of the
sign vector.  We split the boundary faces of all selected
boxes until each face $x$ of a selected box $a$ is \emph{either} contained in some non-selected box \emph{or} the intersection
of $x$ with each non-selected box is at most $(d-2)$-dimensional. 
For each face $x$ of a selected box $a$ that is a sub-face of some non-selected box $b$, 
we delete the $l$-th entry from the sign vector of $b$ and assign this as a new sign vector to
$x$. The list of all such oriented $(d-1)$-boxes and their sign vectors is
denoted by $\mathit{faces}$. This is a sufficient sign list wrt. $f_{\neg l}$ and $s\,(-1)^{l+1}\,{\rm Deg}(\mathit{faces})$ is returned.

The choice of $l$ and $s$ has no impact on the correctness of the algorithm but
can optimize its speed. We choose $l\in \{1,2,\ldots, d+1\}$ and $s\in \{+,-\}$
in such a way that the number of \emph{selected boxes} is minimal. See Section~\ref{sec:experimental-results} for a
more detailed discussion of this issue.
The algorithm for calculating 
$\deg(id,[-1,1]^2,0)$, $l=1$ and $s=1$ is displayed in
Figure~\ref{fig:identity_degree}.   
\begin{figure}[htbp]
  \centering
   \includegraphics[width=8cm]{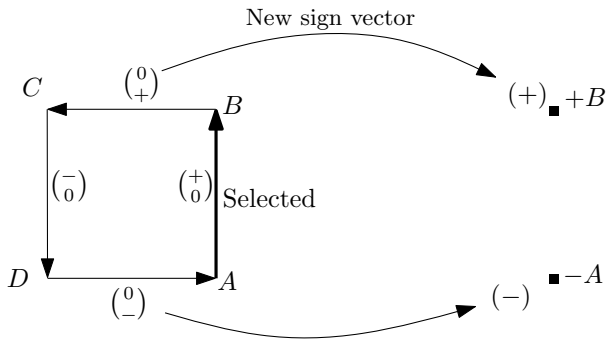}
   \caption{Description of the recursive step for the identity function on the
     oriented box $([-1,1]^2, +)$. 
   For the choice $l=1$ and $s=+$, we have one selected box $AB$.  The degree functions returns 
   ${\rm Deg}(((B,1),(+)), \, ((A, -1), (-)))=\frac{1+(-1)(-1)}{2}=1$ (in this notation, $(B,1)$ is an oriented zero-dimensional box and $(+)$ its sign-vector). From the boxes $BC$ and $DA$, only the sign information is used and the box $CD$ is ignored.}
  \label{fig:identity_degree}
\end{figure}

If the input of the Deg function is a sign list representing a boundary of an oriented cubical set, then the list of selected boxes is exactly the set $\Lambda_{l,s}$ from
Lemma \ref{lemma:1}. 
We will prove in Section \ref{sec:proof-correctness} that the list $\mathit{faces}$ can be subdivided into
$\sum_j \partial D_j + \ud$ where $\{D_j\}_j$ is the decomposition of $\Lambda_{l,s}$ into oriented cubical sets and $\ud$ contains each box $x$ 
the same number of times as $-x$, where $-x$ represents the box $x$ with opposite orientation. 
We will prove that
${\rm Deg}(\mathit{faces})={\rm Deg}(\sum_j {\partial D_j})=\sum_j {\rm Deg}(\partial D_j)$. Together with equation $(\ref{degreeformula})$ this implies 
$${\rm Deg}(L)=s\,(-1)^{l+1}\,\sum_j {\rm Deg}(\partial D_j)=\deg(f,\Omega,0).$$
One example of a possible $\mathit{faces}$ construction is displayed in Figure \ref{fig:faces}.

\begin{figure}[htbp]
  \centering
   \includegraphics[width=7cm]{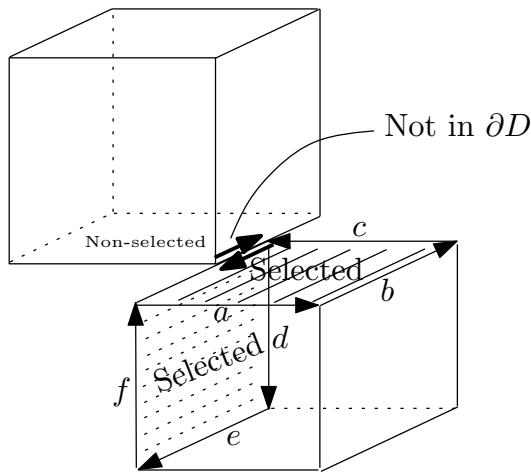}
   \caption{In this case, $B$ is an oriented cubical set containing two $3$-dimensional boxes and the Deg function input
$\partial B$ consists of twelve $2$-dimensional boxes. Two of them are \emph{selected} and form an oriented cubical set $D$ 
with oriented boundary $\partial D=\{a,b,c,d,e,f\}$. The list \emph{faces} contains two more boxes, identical with opposite orientation.}
  \label{fig:faces}
\end{figure}

\subsection{Pseudocode}
\bigskip

\begin{tabbing}\hs\hs\=\hspace*{1cm}\=\hs\hs\=\kill
{\bf Function} Main\\
\textbf{Input}: \\
\> $B$:\> oriented $n$-box \\ 
\> $I(f)$: \> algorithmic representation of an interval-computable $\R^n$-valued function $f$\\
\>\> s.t. $0\notin f(\partial B)$\\
\textbf{Output}: the degree $\deg(f,B,0)$\\
\\
%ob $\leftarrow$ a list containing the $2n$ oriented faces of $B$\\
boundary\_info $\leftarrow$ refineCov($I(f)$, $B$)\\
\textbf{return} Deg(boundary\_info)\\
\end{tabbing}

For an interval-computable function $f$ (see Definition~\ref{def:intcomp}) and
box $a$, if $0\not\in I(f)(a)$, we can infer a non-zero
sign vector entry (see Definition~\ref{def:signcover}) for $a$. Moreover, due to
interval-computability, if $a$ is small enough and $0\notin f(a)$, then $0\notin I(f)(a)$.
Hence, a function with the following specification can be easily implemented by starting with the list of $2n$ faces of $B$, using
$I(f)$ to assign sign vectors to them so that the constructed sign list is wrt. $f$, 
and recursively splitting the boxes in the list 
until the interval evaluation $I(f)$ computes the necessary sign information for it to be sufficient. 

\begin{tabbing}\hs\hs\=\hspace*{1cm}\=\hs\hs\=\kill
{{\bf Function} refineCov}\\
\textbf{Input}:\\
\> $B$: \> an $n$-box in $\R^n$\\
\> $I(f)$: \> algorithmic representation of an interval-computable $\R^n$-valued function \\
\>\> s.t. $0\notin f(\partial B)$\\
\textbf{Output}:\\
\> Sufficient sign list wrt. $f$, covering the oriented boundary $\partial B$ of $B$. \\
\end{tabbing}

Now, finally, we can compute the degree from a sufficient sign covering. 
%Here, we use $+$ to denote the concatenation of lists.

\begin{tabbing}\hs\hs\=\hs\hs\=\hs\hs\=\hs\hs\=\kill
{\bf Function} Deg\\
\textbf{Input}: $L$: Sufficient sign list wrt. some function $f$, covering the oriented boundary $\partial B$ of $B$. \\
%\hspace*{1cm} an oriented boundary $\partial B$ of a cubical set $B$, wrt. $f: B\to \R^{d+1}$\\
\textbf{Output}: $\deg(f,B,0)$ \\
 \textbf{if} $d=0$ \textbf{then}\\
\> \textbf{return} $\frac{1}{2}\sum_{(a,sv)\in L} orientation(a)\times sv$\\
 \textbf{else if} $L=\{\}$ \textbf{then}\\
\>\textbf{return} $0$\\
 \textbf{else}\+\\
 \textbf{let} $1\leq l\leq d+1$ and $s\in \{+,-\}$\\ %such that the $|\{a\in L;\,signvector_l(a)=s\}|$ is minimal. // ``strategy choice''
 $L^{sel}\leftarrow \{ (a, (sv_1,\dots,sv_{d+1}))\in L \mid sv_l=s\}$\\
 $L^{non}\leftarrow L-L^{sel}$\\
 $\mathit{faces}\leftarrow \{ \}$\\
 \textbf{for all} $a\in L^{sel}$ \textbf{do}\+ \\
   bound $\leftarrow$ list of the oriented faces of $a$\\
   split the boxes in bound until for all $b\in bound$, either\\
   -- $b$ is a subset of some element of $L^{non}$, or\\
   -- the intersection of $b$ with any element of $L^{non}$ has dimension smaller than $d-1$\\
   \textbf{for every} $b\in bound$ that is a subset of some box $S$ in $L^{non}$ \textbf{do}\+\\
%   \textbf{let} $S$ be an element in $L^{non}$ such that $b$ is a subset of $S$\\
   $\mathit{faces}\leftarrow \mathit{faces} + (b, sv)$, where\\
   \> $sv$ is the sign vector of $S$ with omitted $l$-th component\-\-\\
 \textbf{return} $s (-1)^{l+1}$ Deg($\mathit{faces}$)\-\\
\end{tabbing}
Note that the input/output specification of the function Deg describes the
behavior for calls from the outside. Recursive calls of the function Deg might
violate the condition on the input---it might be a more complicated sign
list. We will discuss details on the structure of that list and correctness of
recursive calls in the following section.

\subsection{Proof of Correctness}
%\footnote{Motivate all definitions and relate
%to the intuition given at the end of Section 3.1 (where does the $\ud$ show up?))}

\label{sec:proof-correctness}
The algorithm first creates a sufficient sign list wrt. $f: \partial B\to \R^{n}$ where $B$ is the input box. 
This sign list is then an input for the recursive function $\rm{Deg}$.
We want to prove that if $L$ is a sufficient sign list wrt. $f$ covering the boundary $\partial B$ of a box $B$,
then ${\rm Deg}(L)$ returns the degree $deg(f,B,0)$.

To prove this, we will analyze the $\rm{Deg}$ function body.
When dealing with $d$-dimensional sufficient sign lists, we always assume that some $l\in\{1,\ldots, d+1\}$
and $s\in \{+,-\}$ has been chosen.  
Let $L$ be a sufficient sign list wrt. $f$. We denote $L^{sel}:=\{a\in L|\,(SV_a)_l=s\}$ and $L^{non}:=L-L^{sel}$
the sub-list of selected and non-selected boxes.
For each $a\in L^{sel}$,
the $\rm{Deg}$ function refines the boundary $\partial a$ until each $x\in \partial a$ is either a subset of some $S\in L^{non}$
or has at most $(d-2)$-dimensional intersection with each $S\in L^{non}$. 
For each $x\in \partial a$ that is a subset of a $S\in L^{non}$, it assigns to $x$
the sign vector $SV_S$ with deleted $l$-th coordinate. We denote the sub-list of all such $x$ constructed from $a$ by 
$\mathit{faces}(a)$.  The list $\mathit{faces}$ constructed in the $\rm{Deg}$ function body satisfies 
$$\mathit{faces}=\sum_{a\in L^{sel}}\,\mathit{faces}(a)$$
and $s\,(-1)^{l+1}\,{\rm Deg}(\mathit{faces})$ is returned.

In this section, we will suppose that some implementation of the algorithm is given. This includes a rule for the choices of
$l,s$, subdivision of the boundary faces of the \emph{selected boxes}, order of the lists $L^{sel}$ and $L^{non}$ and the choice of $S$.
We will show that if the sign list satisfies a certain regularity condition defined in Definition~\ref{def:star}, then the $\rm{Deg}$
function output is invariant with respect to some changes of the input list, including any change of order, merging and splitting some boxes
or adding and deleting a pair of identical boxes with opposite orien\-ta\-tion. This is shown in Lemma \ref{equiv}.
Further, we show that the list $L^{sel}$
can be decomposed into the sum of oriented cubical sets $D_1,\ldots, D_m$ such
that $0\notin \bigcup_{i\in\{1,\dots, m\}}f_{\neg l}(\partial D_j)$ 
and such that the list $\mathit{faces}$ constructed in the $\rm{Deg}$
function body is a merging of $\sum_j \partial D_j$ and a set of pairs $\{x, -x\}$, so that ${\rm Deg}(\mathit{faces})={\rm Deg}(\sum_j \partial D_j)$.
In Theorem \ref{Theorem:correctness} we prove that ${\rm Deg}(\sum_j \partial D_j)=\sum_j {\rm Deg}(\partial D_j)$ and combining this with 
equation $(\ref{degreeformula})$ in Theorem \ref{thetheorem}, we show that if $L$ is a sufficient sign list wrt. $f$ covering the boundary $\partial B$ of a box $B$,
then ${\rm Deg}(L)$ returns the degree $\deg(f,B,0)$.
\begin{definition}
\label{def:lists-equivalence}
Let $L$ and $L'$ be two sufficient sign lists wrt. $f$. We say that 
$L$ is equivalent to $L'$ and write $L\simeq L'$, if $L'$ can be created from $L$ by applying a finite number of
the following operations:
\begin{itemize}
\item Changing the order of the list,
\item Replacing some oriented box $a$ in one list by two boxes $a_1, a_2$ where $a_1, a_2$ is the splitting of $a$
into two oriented sub-boxes with equal sign vectors $SV_a=SV_{a_1}=SV_{a_2}$,
\item Merging two oriented boxes $a_1$, $a_2$, that form a splitting of some box $a$ and have the same sign vector $SV_{a_1}=SV_{a_2}$,
to one list element $(a,SV_{a_1})$,
\item Adding or deleting a pair of oriented boxes $a$ and $-a$ where $-a$ is the box $a$ with opposite orientation (the sign vectors
$SV_a$ and $SV_{-a}$ do not have to be necessary equal in this case),
\item Changing the sign vectors of some oriented boxes so that the sign covering is still sufficient and wrt. $f$.
\end{itemize}
\end{definition}
Clearly, $\simeq$ is an equivalence relation on sign lists.

\begin{definition}
\label{def:star}
Let $L$ be a $d$-dimensional sufficient sign list wrt. $f$. 
We say that $L$ is \emph{balanced}, if each sub-face $x$ of some $a\in L$ such that
for each $b\in L$, either $x\subseteq b$, or $x\cap b$ is at most $(d-2)$-dimensional \footnote{Here
$x$ and $b$ represent just the box, without taking care of the orientation.}, satisfies
$$|S_x|=|S_{-x}|$$ where $S_x$ is a sub-list of $L$ containing all $a$ s.t. $x$ is an oriented sub-face of $a$. 

In other words, $x$ is a sub-face of some oriented box in $L$ the same number of times as $-x$.
\end{definition}

A sign list representing the oriented boundary $\partial B$ of an $n$-box $B$ is clearly balanced, because for each
$(n-2)$-dimensional sub-face $x$ of some $a\in\partial B$ that is small enough to have either lower-dimensional 
or full intersection with each $b\in\partial B$, $x$ is an oriented sub-face of exactly one $a\in\partial B$ and $-x$ 
is an oriented sub-face of exactly one $a'\in\partial B$.
The following Lemma says that
the property of being balanced is also preserved in the $\mathit{faces}$ construction procedure.
This implies that all input lists $L$ within the recursive $\rm{Deg}$ function are balanced.

\begin{lemma}
\label{lemma:star}
Let $L$ be a sufficient sign list wrt. $f$ that is balanced. Then the list $\mathit{faces}(L)$
created in the Deg function body is also balanced.
\end{lemma}
The proof of this is technical and we postpone it to the appendix. 

\begin{lemma}
\label{equiv}
Let $L$ be a balanced sufficient sign list wrt. $f$ and $L'$ be equivalent to $L$. Then ${\rm Deg}(L)={\rm Deg}(L')$.
\end{lemma}
\begin{proof}
We prove this by induction on the dimension of the sign list. If $L$ is a $0$-dimensional sign list, then
nontrivial merging and splitting of a box is impossible. Independence of order of the list follows from the formula
${\rm Deg}(L)=\frac{1}{2}\,\sum_{a\in L}\, orientation(a)\times\,SV_a$ and adding a pair $(x,-x)$ to the list will
add to the sum $\frac{1}{2}\,(SV_x-SV_{-x})=\frac{1}{2}(\sign(f(x))-\sign(f(x)))=0$.

Assume that the lemma holds up to dimension $d-1$. 
Let $L'$ be a permutation (i.e. the same multiset with different order of elements) 
of a $d$-dimensional sign list $L$ and $\mathit{faces}'$ be the list created for $L'$ in the $\rm{Deg}$ function body.
Changing the order of the list possibly changes the order of 
$L^{sel}$ and $L^{non}$. However, $a\in L^{sel}$ if and only if $a\in (L')^{sel}$ and the same number of times.
Further, $\mathit{faces}(a)$ and $\mathit{faces}'(a)$ can be constructed from each other by a finite number of splitting, merging and
sign vector changing operations, because both are sufficient sign list wrt. $f_{\neg l}$ representing a sign covering of the set
$$\cup \,\{x\,|\,x\,\,{\rm is\,\,a\,\,boundary\,\,sub\mathrm{-}face\,\,of\,\,}a\,\,{\rm and}\,\,x\subseteq n\,\,{\rm for\,\,some}\,\,n\in L^{non} \}.$$
So, $\mathit{faces}'\simeq \mathit{faces}$ and ${\rm Deg} \,L=s\,(-1)^{l+1}\,{\rm Deg}\, \mathit{faces}\,=s\,(-1)^{l+1}\,\,{\rm Deg}\, \mathit{faces}'\,=\, {\rm Deg}(L')$.

Further, let $L'$ be created from $L$ by splitting or merging some oriented box and $\mathit{faces}'$, resp. $\mathit{faces}$ the list constructed in the Deg function body. 
If we split or merge a non-selected box,
then $\mathit{faces}'$ will be equivalent to $\mathit{faces}$, because the equivalence class of $\mathit{faces}(a)$ depends only on the union
of all non-selected boxes. Splitting a selected box $a$ into $a_1, a_2$ will result in splitting some elements of $\mathit{faces}(a)$,
possibly changing their sign-vectors (depending on the choice of $S$ in the algorithm) compatibly with $f_{\neg l}$
and generate a finite number of new pairs $e$ and $-e$ s.t. $e\in \mathit{faces}'(a_1)$ and $-e\in \mathit{faces}'(a_2)$.  
So, $\mathit{faces}$ is again equivalent to $\mathit{faces}'$ and we can apply the induction.

Assume that we change the sign vector of an element in $L$ in such a way that we still have a sufficient sign list wrt. $f$.
If we change the sign vector of a box such that we don't change a selected box to a non-selected or vice versa,  
then this change may only result in a possible change of sign vectors in $\mathit{faces}$ wrt. $f_{\neg l}$ (and possibly splitting and merging
of the boxes in $\mathit{faces}$, if the sign vector change has an impact on the choice of $S\in L^{non}$ in the algorithm). 
So, in this case, $\mathit{faces}\simeq \mathit{faces}'$.
Assume that we change the sign vector $SV_a$ so that
some $a\in L^{non}$ will become selected. Denote $L$ to be the original sign list ($a\in L^{non}$) and $L'$ to be the new sign list
in which $a\in {L'}^{sel}$ and let $\mathit{faces}$, resp. $\mathit{faces}'$ be the corresponding sign lists created in the $\rm{Deg}$ function body.
First note that the sublists of $\mathit{faces}$ containing all elements that are not sub-faces of $a$ and the sublist of $\mathit{faces}'$ containing
all elements that are not sub-faces of $a$, are equivalent, so we only have to analyze the changes caused by the changed sign-vector of $a$.
We claim that the sign list $\mathit{faces}'$ is equivalent to $\mathit{faces}+\partial a$, where 
$\partial a$ is a sign list covering a boundary of $a$ such that all $x\in\partial a$ are endowed with the old sign vectors $SV_a$
with $l$-th entry deleted. 
An implementation of the Deg function body will create, in the $\mathit{faces}'(a)$ construction, 
a decomposition $\partial a=a^{sel}\cup a^{non}$,  where each oriented box in $a^{sel}$
has at most $(d-2)$-dimensional intersection with each $b\in {L'}^{non}$ and each oriented box in $a^{non}$ 
is a subset of some $b\in {L'}^{non}$. 
It follows that each $x\in a^{non}$ is contained in $\mathit{faces}'(a)$ and the list $\mathit{faces}'$ contains $x$ one more time than $\mathit{faces}$.
Further, due to the fact that $L$ is balanced, 
for each $x\in a^{sel}$, there exist the same number
of boxes $u$ in $L$ s.t. $x$ is an oriented sub-face of $u$ as boxes $v$ s.t. $-x$ is an oriented sub-face of $v$, 
$a$ being among the $u$'s. All such $u$ and $v$'s are in ${L'}^{sel}$, $a$ being the only of these boxes contain in $L^{non}$. 
This implies that the list $\mathit{faces}$ is equivalent to a list containing each such $-x$ one more time than $x$.
After deleting a finite number of pairs $(x,-x)$, $\mathit{faces}$ is equivalent to a list containing one $-x$ for each $x\in L^{sel}$
(it comes from $\mathit{faces}(v)$ for some $v\in L^{sel}$ containing a sub-face of $a\in L^{non})$. 
In $\mathit{faces}'$, there is no such $-x$, because $x$ is not contained in any $b\in {L'}^{non}$. 
Summarizing this, $\mathit{faces}'$ can be constructed from $\mathit{faces}$
be adding a sign list covering $|a^{non}|$ and deleting a sign list covering $|a^{sel}|$. This is equivalent to adding a sign list
covering all $|\partial a|$ and we obtain that $\mathit{faces}'\simeq \mathit{faces}+\partial a$. By induction, we may assume that all boxes in $\partial a$
has equal sign vector, compatible with $f_{\neg l}$.
Now we need to show that adding the full boundary $\partial a$ of $a$ endowed with a constant sign vector 
does not change the $\rm{Deg}$ output.
In the $0$-dimensional case, this says that ${\rm Deg}(L+\partial a)={\rm Deg}(L)+\frac{1}{2}\,(s-s)$ where
$s$ is the sign of $f$ on $a$. 
Let $L'=L+\partial a$ be a sign list of positive dimension such that all elements in $\partial a$ are endowed with the same sign-vector.
In the consequential $\mathit{faces}$ construction, either all boundary faces of $a$ will be \emph{selected} 
or they will be all \emph{non-selected}. In the first case,
$\mathit{faces}'$ will be a sum of $\mathit{faces}$ and pairs $(x,-x)$. In the second case, 
$\partial a$ may be refined so that each element is either a subset of some other non-selected box, 
or has only lower-dimensional intersection with each non-selected box.
Those $\alpha\in\partial a$ that are a subset of some other non-selected box can only possibly change the sign vector
of some boxes in $\mathit{faces}$. Boxes $\beta\in\partial a$ that have only lower-dimensional intersection with each non-selected box
will lead (after possibly merging and splitting the $\mathit{faces}$ list) to the addition of a sum of pairs $x$ and $-x$ to $\mathit{faces}$ due to the fact
that $\mathit{faces}$ is a balanced sign list.  So, $\mathit{faces}'\simeq \mathit{faces}+\partial a\simeq \mathit{faces}$ and ${\rm Deg}(L')={\rm Deg}(L)$.

Finally, adding a pair of two selected boxes $a$ and $-a$ will create additional pairs $x$ and $-x$ in $\mathit{faces}$.
Adding a pair of two non-selected boxes $a$ and $-a$ may enlarge the union of the non-selected boxes. Let $L':=L+a+(-a)$ for some
non-selected $a$. The $\mathit{faces}'$ list created in the $\rm{Deg}$ function body is equivalent (after merging and splitting some boxes)
to a sum $\mathit{faces}'_1 + \mathit{faces}'_2$, where $\mathit{faces}'_1$ consists of all oriented sub-faces $x$ of some $a\in {L'}^{sel}$ that are contained in some
$b\in L^{non}$ and $\mathit{faces}'_2$ consists of all oriented sub-faces $x$ of some $a\in {L'}^{sel}$ that are contained in $a$ but have 
at most $d-2$-dimensional intersection with each $b\in L^{non}$. We may further split the boxes in $\mathit{faces}'_2$ and suppose that for each 
$x\in \mathit{faces}'_2$ and $b\in L$, either $x\subseteq b$ or $x\cap b$ is at most $d-2$-dimensional.
Then $\mathit{faces}\simeq \mathit{faces}'_1$ and due to the balancedness of
$L$, each $x\in \mathit{faces}'_2$ is a~sub-face of some
$u\in L$ the same number of times as $-x$ is a sub-face of some $v\in L$. All these $u$ and $v$'s have to be in $L^{sel}$, because
$x$ has a lower-dimensional intersection with each $b\in L^{non}$.
So, the $\mathit{faces}'_2$ list is equivalent to a sum of pairs $(x,-x)$ and $\mathit{faces}'\simeq \mathit{faces}'_1 + \mathit{faces}'_2\simeq \mathit{faces}$.
If we add two boxes $a$ and $-a$ such that $a$ is selected and $-a$ non-selected,
we may change the sign vector of $-a$ (due to the previous paragraph) so that both $a$ and $-a$ are selected and the $\rm{Deg}$ function output doesn't change. 
\end{proof}

\begin{theorem}
\label{Theorem:correctness}
  Let $B$ be an $n$-box and $I(f)$ be an algorithm representing an interval-computable function $f: B\to\R^n$ s.t. $0\notin f(\partial B)$.
  The presented algorithm, run with $B$ and $I(f)$ as inputs, terminates and returns the degree $\deg(f,B,0)$.
\end{theorem}

\begin{proof}

The theorem is a consequence of statement~\ref{part_two} of the following:
\begin{enumerate}
\item\label{part_one} Let $\Omega_1, \ldots \Omega_k$ be oriented cubical sets of dimension $d+1$, let $L_1, \ldots, L_k$ be $d$-dimensional sufficient sign lists 
wrt. to a function $f: \cup|\Omega_i| \to \R^{d+1}$ s.t. the boxes in $L_i$ are $d$-boxes forming an oriented boundary $\partial\Omega_i$ 
of $\Omega_i$ for all $i$.
Then ${\rm Deg}(\sum_i {L_i})=\sum_i {\rm Deg}(L_i)$. 
\item\label{part_two}
Let $\Omega$ be a $(d+1)$-dimensional oriented cubical set and 
let $L$ be a $d$-dimensional sufficient sign list wrt. a function $f: |\Omega|\to \R^{d+1}$, such that the boxes in $L$ form an oriented boundary of $\Omega$.
Then ${\rm Deg}(L)$ returns the number $\deg ({f}, \Omega, 0)$. 
  \end{enumerate}

We prove both statements~\ref{part_one} and~\ref{part_two} by induction on the dimension $d$. 
If the sign lists are $0$-dimensional, then ${\rm Deg}(L)=\frac{1}{2}\sum_{a\in L} orientation(a)\times SV_a$ where $SV_a$ is the $1$-dimensional 
sign vector of $a\in L$ and ${\rm Deg}(\sum_i L_i)=\sum_i {\rm Deg}(L_i)$ is true for any sufficient $0$-dimensional sign lists $L_i$. 
For statement~\ref{part_two}, the Deg function result is compatible with the one-dimensional formula 
$$\deg(f, \overrightarrow{ab}, 0)=\frac{1}{2}(\sign f(b)-\sign f(a))$$ for $f: \overrightarrow{ab}\to\R$.

Assume that the dimension is $d>0$ and both \ref{part_one} and \ref{part_two} hold for lower-dimensional sign lists. 
First we prove \ref{part_two}. Let $L$ be a sufficient sign list such that its oriented boxes 
form the boundary $\partial\Omega$ of a $d+1$-dimensional oriented cubical set $\Omega$. We know that $L$ is balanced. 
Let $l\in\{1,\ldots,d+1\}$ and $s\in\{+,-\}$ be chosen in the Deg function body.
For each box $a\in L$, choose an index $l(a)$ s.t. 
\begin{itemize}
\item if $(SV_a)_l=s$, then $l(a)=l$ and $s(a)=s$
\item if $(SV_a)_l\neq s$, then choose $l(a)$ and $s(a)$ so that the sign vector $(SV_a)_{l(a)}=s(a)\neq 0$ 
\end{itemize}
Such index $l(a)$ and sign $s(a)$ exist for each $a$, because the sign list is sufficient.
For each $l'\in\{1,\ldots,d+1 \}$ and $s'\in\{+,-\}$,
denote $\Lambda_{l',s'}$ a list of all boxes in $a\in L$ such that $l(a)=l'$ and $s(a)=s'$.
The list of selected boxes $L^{sel}$ is formed exactly by the boxes in $\Lambda_{l,s}$ and the list of non-selected boxes
is $L^{non}:=L-L^{sel}$. 
It follows from Lemma~\ref{lemma:1} that the there exist $(d-1)$-dimensional cubical sets $D^j_{l',s'}$
such that $\cup_{j}\,D^j_{l',s'}=\Lambda_{l',s'}$ holds for all $l'\in\{1,\ldots,d+1\}$ and $s'\in\{+,-\}$.
%and that there exist oriented boundaries $\partial D^j_{l',s'}$ such that each $x\in \partial D^j_{l',s'}$ is a sub-face
%of some $a\in \Lambda_{l_1, s_1}$ for some $l_1\neq l'$.
For each $j$ and each $a\in D^j_{l,s}$, let $\mathit{faces}(a)$ be the $(d-1)$-dimensional sign list created from
the sub-faces of $a$ in the $\rm{Deg}$ function body, and let $\mathit{faces}=\sum_{a\in L^{sel}} \mathit{faces}(a)$.
Let $\mathit{faces}(a)_{split}$ be a splitting of $\mathit{faces}(a)$ such that for each 
$e\in \mathit{faces}(a)_{split}$ and each $b\in\partial\Omega$,
either $e\subseteq b$ or $e\cap b$ is at most $(d-2)$-dimensional. Further, define $\partial D^j_{l,s}$
to be the sub-list of $\sum_{a\in D^j_{l,s}} \mathit{faces}(a)_{split}$ containing all $x$ such that
there exists a \emph{unique} $a\in D^j_{l,s}$ s.t. $x$ is a sub-face of $a$ (we don't take care of orientation here). 
This is a sign list covering an oriented boundary of $D^j_{l,s}$ (see Def. \ref{def:oriented_boundary}).

Define $\mathit{faces}_{split}:=\sum_{a\in L^{sel}} \,\mathit{faces}(a)_{split}$. 
Let $x\in \mathit{faces}_{split} - \sum_j\,\partial  D^j_{l,s}$
and assume that $x\in \mathit{faces}(S)_{split}$ for $S\in D^j_{l,s}$.
Because $x\notin\partial D^j_{l,s}$, $x$ is a sub-box of exactly two boxes $S$ and $S'$ in $D^j_{l,s}$ and $x\subseteq b$
for some non-selected box $b$. The sub-list $\mathit{faces}(S')_{split}$ contains a box $y$ s.t. $y\cap x$ is
$(d-1)$-dimensional ($y$ is a sub-box of some face $e$ of $S'$ and $e\cap b$ is $(d-1)$-dimensional). 
The orientation of $y$ induced from $S'$ is different from the orientation of  $x$ (see Def. \ref{oriented_cubical_set}).
So, after possible further splitting of the list $\mathit{faces}_{split}$, we may assume that $y=-x$ and that for each 
$j$, $\sum_{a\in D^j_{l,s}} \mathit{faces}(a)_{split}$ contains either both $x$ and $-x$ or none of them. It follows that
the list $\mathit{faces}_{split}$ contains $x$ the same number of times as $-x$ and the list $\mathit{faces}$ is equivalent to $\sum_j \partial D^j_{l,s}$.
Now we derive
\begin{eqnarray*}
&&{\rm Deg}(L)=s\,(-1)^{l+1}\,{\rm Deg}(\mathit{faces})= \,\,{\rm \,(Lemma\,\, \ref{equiv}})\,\,= s\,(-1)^{l+1}\,{\rm Deg}(\,\sum_j \,\partial D^j_{l,s}) =\\
&& \,{\rm (Induction,\,\,\ref{part_one}.)}\,\,=s\,(-1)^{l+1}\,\sum_j\,{\rm Deg}(\partial D^j_{l,s}) =\,{\rm (Induction,\,\,2.)}=\, \\
&& =s\,(-1)^{l+1}\,\sum_j deg(f_{\neg l}, D^j_{l,s}, 0)=s\, (-1)^{l+1}\sum_{j;\,\,l'=l\,\,{\rm and}\,\,s'=s} deg(f_{\neg l}, D_{l',s'}^j,0)=\\
&& {\rm (Theorem\,\,\ref{thetheorem},\,\,equation \,\,(\ref{degreeformula}))}=\deg(f,\Omega,0).
\end{eqnarray*}

It remains to prove \ref{part_one}. 
Assume that statement \ref{part_one} holds up to dimension $d-1$, and \ref{part_two} holds up to dimension $d$. 
Let $L=\sum_i L_i$, $L_i$ be a $d$-dimensional sufficient sign lists wrt. $f$ such that the boxes in $L_i$ form oriented boundaries $\partial\Omega_i$ 
of oriented cubical sets $\Omega_i$ for $i=1, \ldots, k$.

In the same way as before, we define for $i=1,\ldots, k$ the sets $D(i)^j_{l',s'}$ to be oriented cubical sets such that 
$L_i$ is the disjoint sum $\sum_{j,l',s'}\,D(i)^j_{l',s'}$, the sign vectors have $l'$th component $s'$ on $D(i)^j_{l',s'}$
and the oriented boundaries $\partial D(i)_{l,s}^j$ are sub-lists of a splitting of $\mathit{faces}(L)$ 
such that each $x\in\partial D(i)^j_{l,s}$ is a sub-face of some $b\in D(i)^{j'}_{l',s'}$ for some $l'\neq l$.
Similarly as before, $\mathit{faces}$ is a equivalent to $\sum_{i,j}\,\partial D(i)^j_{l,s}$ and
\begin{eqnarray*}
&&{\rm Deg}({L})=s\, (-1)^{l+1}\,{\rm Deg}(\mathit{faces})= \,{\rm \,(Lemma\,\,\ref{equiv})\,} =\,s\, (-1)^{l+1}\,{\rm Deg}(\sum_{i,j}\partial D(i)^j_{l,s})=\\
&& {\rm \,(Induction, \,\,\ref{part_one}.)\,} = s\,(-1)^{l+1}\,\sum_{i,j} {\rm Deg}(\partial D(i)^j_{l,s})=\,\,{\rm (Induction,\,\, \ref{part_two}.)}\, \\ 
&& =s\,(-1)^{l+1}\,\sum_{i,j} deg({f}_{\neg l}, D(i)_{l,s}^j, 0)=s\,(-1)^{l+1}\,\sum_{i,j;\,l'=l\,\,{\rm and}\,\,s'=s} deg(f_{\neg l}, D(i)_{l',s'}^j,0)=\\
&&  \,\, {\rm (Equation \,(\ref{degreeformula}))}=\sum_i \deg(f,\Omega_i,0)=\,\,({\rm Statement\,\,2.})\,\,=\sum_i {\rm Deg}(L_i)
\end{eqnarray*}
which completes the proof.  
\end{proof}

From this proof it can be seen that our approach to degree computation is not
restricted to boxes, but works for general cubical sets: in
Item~\ref{part_two} of this proof, we showed that ${\rm Deg}(L)$ returns the degree
$deg(f,\Omega,0)$, if $L$ is a $d$-dimensional sufficient sign list wrt. a
function $f: |\Omega|\to \R^{d+1}$, such that the boxes in $L$ form an oriented
boundary of $\Omega$. So, for a function $f$ defined on a $(d+1)$-dimensional
cubical set $|\Omega|$ embedded in $\R^n$ s.t. $0\notin f(\partial |\Omega|)$,
we might algorithmically find a subdivision of the oriented boundary
$\partial\Omega$, create a sufficient sign list $L$ wrt. $f$ and run ${\rm Deg}(L)$.
%We assumed that the domain is a box only in order to keep the paper simple, but
%the algorithm works for functions defined on general cubical sets. 

\section{Experimental Results}
\label{sec:experimental-results}

We tested a simple implementation of the algorithm on several algebraic functions $f$ and boxes
$B$. All timings were measured running version 1.0 of the implementation on a PC with Intel Core i3 2.13
GHz CPU and 4GB RAM. Interval computations were done based on the library
smath~\cite{Hickey} implementing intervals with floating point endpoints and
conservative rounding. In theory it could happen that 64 bit floating point
representation does not suffice for computing a sufficient sign covering of $\partial B$, but
in our experiments we did not find a single example where this happened.

Unfortunately, up to the best of our knowledge, all
published articles on general degree computation algorithms, only contain
examples of low dimension, for which our algorithm tends to terminate with a
correct result in negligible time. Hence, in order to show the properties and
limitations of our algorithm, we chose different examples that scale to higher
dimensions.

The first part of the algorithm where boundary boxes are subdivided
and sign vectors are computed takes usually about 5-50 times less
than the combinatorial part where the degree is calculated from the list of boxes and sign vectors. 
However, if there is no solution of
$f(x)=0$ on $B$ (and the degree is zero),  then the second part terminates immediately, because 
---in the simplest case---there are no "selected boxes" at all.

In most cases, computation of $\deg(f,B,0)$ such that $0\in f(B)\setminus f(\partial B)$, terminated in reasonable time if
$\dim\, B\leq 10$.
If $0\notin f(B)$, then the degree is zero and the algorithm terminates very fast even in much higher dimensions.
\\
\\
{\it Example 1.} For the identity function on $[-1,1]^n$, the degree computation terminates even in high dimensions. The times are given in Figure \ref{fig_id}.
\begin{figure}[htbp]
  \centering
   \includegraphics{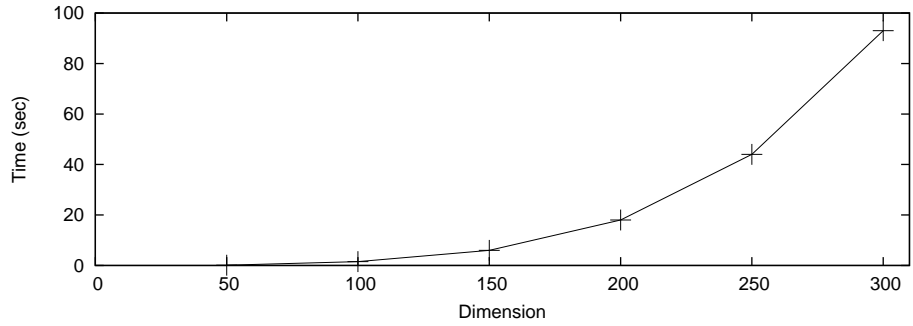}
   \caption{Time spent on calculating $\deg(id,[-1,1]^n,0)=1$.}
\label{fig_id}
\end{figure}
\\
\bigskip \\
{\it Example 2.}
We considered the function 
\begin{eqnarray*}
&& f_1=x_1^2-x_2^2-\ldots -x_n^2\\
&& f_2=2x_1x_2 \\
&& \ldots\\
&& f_n=2x_1 x_n.
\end{eqnarray*}
This function has a single root in $x=0$ of degree $2$ for $n$ even and $0$ for $n$ odd. 
Figure $\ref{fig_deg_two}$ shows the time consumed for calculating $\deg(f,B,0)$ for $B=[-1,1]^n$ and $B=[-0.001,1]^n$. 
The computation is significantly faster for $B=[-\epsilon,1]^n$ where $\epsilon>0$ is small and the root $0$ is close to the boundary. 
%We don't completely understand this phenomenon yet.
In this case, the subdivision of the boundary contains only two selected boxes (both close to $0$).
For~$B=[-\epsilon,\epsilon]^n$, the calculation takes about the same time as for $B=[-1,1]^n$.

\begin{figure}[htb]
   \includegraphics{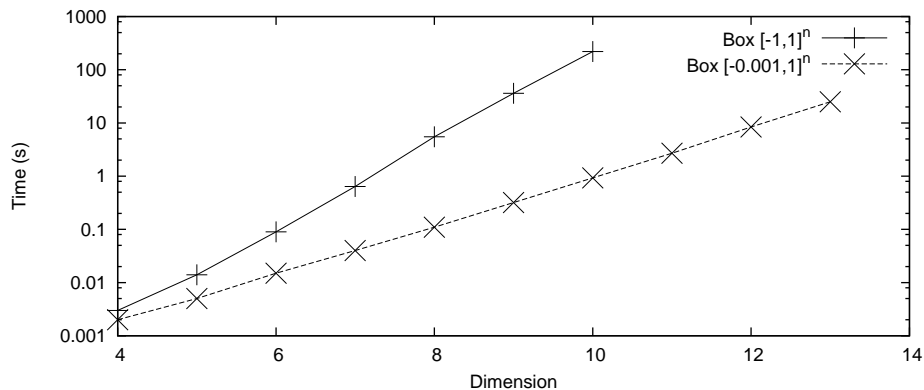}
   \caption{Time needed to calculate the degree $\deg(f,[-1,1]^n,0)$ and
     $\deg(f,[-\epsilon,1]^n,0)$ for $\epsilon=0.001$.}
\label{fig_deg_two}
\end{figure}

The following table shows the number of {\it selected} and {\it non-selected} boxes in the subdivision of $\partial B$ for $B=[-1,1]^n$.

{\begin{center}
\begin{tabular}{l | l | l}
Dim $B$\, & Selected boxes\, & Non-selected boxes\\
\hline
5 & 32 & 800\\
6 & 64 & 2368\\
7 & 128 & 6528 \\
8 & 256 & 17408 \\
9 & 512 & 44032 \\
10 & 1024 & 108544\\ 
\end{tabular}
\end{center}
}

If we chose the box to be $[\epsilon,1]^n$ or any other such that $0\notin f(B)$, the degree calculation terminates almost immediately even in dimension 1000 and more.\\
We also investigated the effect of the choice of $l$ and $s$ in the Deg function
body. By default, they are chosen so that the number of selected boxes is minimal.
Numerical experiments show that the computation takes more time if the number of selected boxes is larger. The following table shows the number of selected boxes for various $l$ and $s$
in dimension $6$.
\\
{
\begin{center}
\begin{tabular}{l  l | l || l l | l}
$l$ & $s$ & Nr. of boxes & $l$ & $s$ & Nr. of boxes \\
\hline
1 & + & 64 & 1 & - & 1600\\
2 & + & 64 & 2 & - & 64\\
3 & + & 64 & 3 & - & 64\\
4 & + & 64 & 4 & - & 64\\
5 & + & 96 & 5 & - & 96\\
6 & + & 96 & 6 & - & 96\\
\end{tabular}
\end{center}
Choosing the bad strategy choice $l=1$ and $s=-$ would increase the computation time significantly. 
The following table shows the time comparison.
{\begin{center}
\begin{tabular}{l | l | l}
Dimension & Optimal choice of $l$ and $s$ & Worst choice of $l$ and $s$\\
\hline
6 & 0.09 s & 0.8 s\\
7 & 0.65 s &  10 s\\
8 & 5.6 s &  175 s\\
\end{tabular}
\end{center}} 

In general, we made the observation that for a fixed choice of $l$
and $s$, for some permutations of variables the number of selected boxes, and
hence run-time, is much higher than for others. Hence, our strategy of choosing
$l$ and $s$ makes the run-time of the algorithm much more robust. \\
\bigskip \\
{\it Example 3.} We also tested the algorithm on the non-Lipschitz function 
$$\sqrt[3]{f}:=(\sqrt[3]{f_1}, \sqrt[3]{f_2},\ldots, \sqrt[3]{f_n}): [-1,1]^n\to\R^n$$
where $(f_1,\ldots, f_n)$ is the function from Example 2. 
The construction of the sign covering of the boundary
takes more time than in the previous example, because more interval computations are involved.
However, the sign covering of the boundary is identical to that in~Example~2, because for all intervals $[a,b]$ that occur in 
this computation, $[a,b]$ doesn't contain $0$ if and only if our implementation of $I(\sqrt[3]) ([a,b])$ doesn't contain $0$
and both intervals have the same sign.
So, the combinatorial part is identical to the previous example. 
We compare the running time of the numerical part of the computation for $f$ and $\sqrt[3]{f}$
in the following table~\footnote{Our implementation 
of $I(\sqrt[3])$ is based on the real number identity $\sqrt[3]{x}=\sign(x)\,\exp(\frac{1}{3}\ln{|x|})$. 
For the absolute value, logarithm and exponentiation we used the interval functions available in the smath library~\cite{Hickey}.}.
{\begin{center}
\begin{tabular}{l | l | l}
Dimension \,\,& \, Sign covering wrt. $f$ \,\, & \, Sign covering wrt. $\sqrt[3]{f}$\\
\hline
8 &  0.2 s &  5.1 s\\
9 &  1 s &  15.4 s\\
10 &  5.3 s &  49.5 s\\
\end{tabular}
\end{center}}

% {\it Example 3.} For linear functions $f(x)=Ax$, the degree $\deg(f,B,0)$ is
% just the sign of $\det(A)$ if $0\in B$. However, our method does not exploit the
% fact that the function is linear and the calculation is slow due to a large
% number of Selected and Non-selected boxes in the subdivision of $\partial B$.
% We calculated the degree of the linear function $f_i(x)=\sum_{j\neq i} x_j$,
% $i=1,\ldots, n$ and dimension 6 it needed already 39 seconds for computing the
% degree. However, pre-conditioning~\cite{Kearfott:96a,Kearfott:91a,Neumaier:90}
% can easily reduce such problems to almost identity functions, for which (see
% Example 1) the algorithm is efficient even in high dimensions. 

% Moreover, for linear functions with a sparse matrix, the calculation
% terminates quicker, even without pre-conditioning.

%\section{Conclusion}
%
%Different algorithms to compute degree from sign covering.

\appendix
\section{Proof of Lemma \ref{lemma:1}}
Let us adopt the notation $a\hr B$ for \emph{``$a$ is an oriented sub-face of $B$''} (see Def. \ref{boundary_orientation}). Let
$\partial\Omega$ be an oriented boundary of the oriented $d$-dimensional cubical set $\Omega$ and let 
$\Lambda_{l',s'}=\{a\in\partial\Omega\,|\,(SV_a)_{l'}=s'\}$.
$\partial\Omega$ is a disjoint union of the sets $\Lambda_{l',s'}$, $(l',s')\in \{1,\ldots,d\}\times\{+,-\}$.

For each $B\in\Omega$, let $\partial B$ be an oriented boundary of $\{B\}$ that contains all $a\in\partial\Omega$ such that $a\hr B$.
Such oriented boundary $\partial B$ can be constructed by completing $\{a\in\partial\Omega \,| \,a\hr B\}$ to a full oriented
boundary of $B$.
Further, for each $a\in \partial\Omega$, let $\partial a$ be an oriented boundary of $\{a\}$ such that for each $x\in\partial a$, 
the following condition is satisfied:
\begin{itemize}
\item for each $B\in\Omega$ and each $b\in\partial B$, either $x\subseteq b$ or $x\cap b$ is at most $(d-3)$-dimensional.
\end{itemize}
Such oriented boundary $\partial a$ can be constructed by splitting the boundary faces of
$a$ as long as some boundary face $x\in\partial a$ has nontrivial $(d-2)$-dimensional intersection with some $b\in\partial B$
for some $B\in\Omega$. 
Denote by $\partial\Lambda_{l,s}$ the set of all boxes $x\in\partial a$ s.t. $a\in\Lambda_{l,s}$ and $x$ is a sub-face 
(not necessary \emph{oriented} sub-face) of some $b\in \Lambda_{l',s'}$
for $(l',s')\neq (l,s)$. 
Finally, for any oriented box $a$, let $-a$ be the same box with opposite orientation. 

For all $a\in\partial\Omega$, $x\in\partial a$ and $B\in\Omega$, either $x\subseteq B$ (this is when $x\subseteq b$ for some $b\in\partial B$)
or $x\cap B$ is at most $(d-3)$-dimensional. If $x\subseteq b\in\partial B$, then there exist unique $b_1, b_2 \in\partial B$ such that $x\hr b_1$ and
$-x\hr b_2$ ($\partial B$ is an oriented cubical set with empty boundary).
Further, note that for each $b\in\partial B$, either $b\in\partial\Omega$ or $b$ has only lower-dimensional intersection with each 
element of $\partial\Omega$ (if $b$ had a $(d-1)$-dimensional intersection with $c\in\partial\Omega$ and $b\neq c$, 
then $c$ would be a~sub-face of $B$ due to the second condition of Def. \ref{def:oriented_boundary}
and $b,c\in\partial B$ would violate the first condition of Def. \ref{def:oriented_boundary}).

Let $l\in\{1,\ldots,d\}$ and $s\in\{+,-\}$.
We construct the sets $D_j$ and $\partial D_j$ inductively by associating the boxes in $\Lambda_{l,s}$ with sets $D_j$. 
Assume that $D_1,\partial D_1$, $\ldots,$ $D_{k-1}, \partial D_{k-1}$ satisfy the following conditions for all $1\leq i,j\leq k-1$:
\begin{itemize}
\item $D_i\subseteq \Lambda_{l,s}$ is an oriented cubical set
\item $D_i\cap D_j=\emptyset$ for $i\neq j$
\item $\partial D_j\subseteq \partial\Lambda_{l,s}$ is an oriented boundary of $D_j$.
\end{itemize}
Let $D_{k}\subseteq\Lambda_{l,s}$ be an oriented cubical set such that $D_k\cap D_i=\emptyset$ for $i<k$. 
Let $\partial D_k$ be an oriented boundary of $D_k$ s.t. $\partial D_k\subseteq \cup_{a\in D_k} \partial a$ 
(such a boundary exists, because $\partial a$ is subdivided fine enough).
If $\partial D_k\subseteq\partial\Lambda_{l,s}$, 
then condition $\ref{lm:4}.$ from the Lemma is satisfied for each $b\in \partial D_k$ and the construction of $D_k$ is completed. 
In such case, if $\cup_{i=1}^k D_i\neq \Lambda_{l,s}$, then we choose some $a\in \Lambda_{l,s}\setminus \cup_{i=1}^k D_i$ and defining
$a\in D_{k+1}$ we start the construction of a new set $D_{k+1}$.

Assume that $\partial D_k\nsubseteq\partial\Lambda_{l,s}$. Then 
there exists some $x\in\partial D_k$,  $x\notin\partial\Lambda_{l,s}$.  Because $x\in\partial D_k$, there exists exactly one 
$a_0\in D_k$ such that $x\hr a_0$ (Def. \ref{def:oriented_boundary}). 
The condition $x\notin\partial\Lambda_{l,s}$ implies that 
the intersection of $x$ with any $b\in\Lambda_{l',s'}$ for $(l',s')\neq (l,s)$ has dimension at most $d-3$.
We assumed that $a_0\in\partial \Omega$, so there exists a unique box $B_1\in \Omega$ such that $a_0\hr B_1$. 
Let us construct a sequence $a_0, a_1, \ldots, a_p$ and a sequence $B_1,\ldots, B_{p}\in\Omega$ of oriented boxes 
such that the following conditions are satisfied for $u=1,\ldots, p$:
\begin{itemize}
\item $x\hr a_{u-1}\hr B_u$ and $a_{u-1}\in\partial B_u$,
\item $(-x) \hr (-a_{u})\hr B_u$ and $(-a_u)\in\partial B_u$,
\item $B_u$ and $B_{u+1}$ have $(d-1)$-dimensional intersection with compatible orientations,
%\item $x'\subseteq a_u$  for all $u$
\item $(-a_p)\in\partial \Omega$.
\end{itemize}
The boxes $B_1$ and $a_0$ have been defined and $x\hr a_0\hr B_1$. Suppose that $x\hr a_{u-1} \hr B_u$.
Let $(-a_u)$ be the unique oriented box in  $\partial B_u$ s.t. $(-x)\hr (-a_{u})$.
If $(-a_u)\in\partial\Omega$, then $u=p$ and we are done. 
Otherwise, the intersection of $(-a_u)$ with each $b\in\partial \Omega$ is at most $d-2$ dimensional and
it follows from Definitions \ref{oriented_cubical_set}  and \ref{def:oriented_boundary} that $(-a_u)$ 
is a common sub-face of two boxes $B_u$ and $B_{u+1}$ in $\Omega$
with compatible orientations. This means that $(-a_u)\hr B_u$ and $a_u\hr B_{u+1}$, so $x\hr a_u\hr B_{u+1}$.
For all $u$, $x\hr a_{u}$, in particular $-x\hr (-a_p)$ and it follows that 
$a_0$ and $(-a_p)$ have compatible orientations. 
We add the box $(-a_p)$ to $D_k$. We will show that this does not violate any of the above assumptions and we 
redefine $\partial D_k$ so that it is an oriented boundary of $D_k$ and $\partial D_k\subseteq \cup_{a\in D_k} \partial a$.

First we show that the sequence $\{(a_u, B_u)\}_u$ terminates, i.e. it is not periodic. Assume that it is periodic and that
$(-a_u)\notin\partial\Omega$ for all $u>0$. Let $p$ be the smallest integer such that $(a_p, B_p)=(a_k, B_k)$ for some $k<p$.
There exists a unique $a_{p-1}$ s.t. $x\hr a_{p-1}\hr B_p$ and
exactly two boxes $B_{p-1}$ and $B_p$ in $\Omega$ containing $a_{p-1}$ as a sub-face, so $(a_{p-1}, B_{p-1})$ is uniquely
determined by $(a_p, B_p)$. If $k>1$, then this implies $(a_{k-1}, B_{k-1})=(a_{p-1}, B_{p-1})$, contradicting the assumption
that $p$ was the smallest such integer. If $k=1$, then $x\hr a_0=a_{p-1}\hr B_1=B_p$ and $a_0$ is a common sub-face of
two elements $B_p$ and $B_{p-1}\in\Omega$ which contradicts $a_0\in\partial\Omega$. We showed that the sequence $\{(a_u, B_u)\}_u$
terminates and we may add $(-a_p)$ to $D_k$. 

Now we show that adding $(-a_p)$ to $D_k$ doesn't violate any assumption of the construction of the sets $D_j$.
Note that $a_p\neq a_0$. If $a_p=a_0$, then we would have $x\hr a_0\hr B_1\in\Omega$ and $-a_0\hr B_p\in\Omega$. This implies that $B_1\neq B_p$,
$a_0\hr B_1$, $(-a_0)\hr B_p$, which contradicts 
and the assumption $a_0\in\partial\Omega$ (Def.~\ref{def:oriented_boundary}).
Further, if $(-a_p)\notin \Lambda_{l,s}$ then $(-a_p)\in\Lambda_{l',s'}$ for some $(l',s')\neq (l,s)$ and $x$ would be a $(d-2)$-dimensional
sub-face of $(-a_p)$, contradicting the assumption $x\notin\partial\Lambda_{l,s}$. 
This proves that $(-a_p)\in\Lambda_{l,s}$. The box $-a_p$ is not in $D_k$ yet, because $x$
is contained in both $-a_p$ and $a_0$ and we assumed that $x\in\partial D_k$.
Also, $(-a_p)$ is not contained in any $D_i$, $i<k$. If $(-a_p)\in D_i$ for $i<k$, then $a_0$ would be added to $D_i$ before, 
constructing the sequence $(-a_p), (-{a}_{p-1}),\ldots, (-{a}_1), (-a_0)$ where $(-x)\hr (-{a}_{v})\hr B_{v}$
and $x\hr {a}_{v-1}\hr B_{v}$ for all $v=p,\ldots, 2,1$. At the end of this sequence, $a_0=-(-a_0)\in\partial\Omega$ would be included into $D_i$, contradicting 
our starting assumption $D_i\cap D_k=\emptyset$. So, adding $(-a_p)$ to $D_k$ doesn't violate
any assumption of the construction. 

Each $x\in\partial D_j$ is a sub-box of some $b\in\Lambda_{l',s'}$ for $(l',s')\neq (l,s)$. 
However, the case $(l',s')=(l,-s)$ is impossible, because $f_l$ cannot have sign 
$s$ on $|D_j|$ and $-s$ on $x\subseteq |D_j|$. So, $l'\neq l$.
In this way, we construct the oriented cubical sets $D_j$ such that $\cup D_j=\Lambda_{l,s}$. 
This can be done for each $l$ and $s$ and the resulting sets $\{D_j\}_j$ satisfy all the requirements. $\square$

\section{Proof of Lemma \ref{lemma:star}}
Assume that $L$ is a balanced $d$-dimensional sufficient sign list wrt. $f$. First we define some additional notation.
We say that an oriented $(d-1)$-box $e$ is \emph{small wrt. $L$}, if for each $F\in L$, either $e\subseteq F$ or $e\cap F$ is at most $(d-2)$-dimensional,
where $e$ and $F$ represent the boxes, without considering the orientation.
Furthermore, we fix the notation $a\hr B$ for \emph{``$a$ is an oriented sub-face of $B$''} 
(with the induced orientation, see Def. \ref{boundary_orientation}) as in~the proof of Lemma \ref{lemma:1},
and the notation $a\subseteq_o b$ for
\emph{``$a$ is an oriented sub-box of $b$''} (Def. \ref{def:oriented_box}). Further, let us represent the
list $L$ as a \emph{set} of pairs $L\simeq \{(E_1,1),(E_2,2),\ldots, (E_{|L|},|L|)\}$, where $E_i$ is the $i$-th element of $L$. 

Let $$\mathcal{A}=(\{e\,|\,\exists (E,i)\in L\,\,e\hr E \,\,{\rm and}\,\,e\,\,{\rm is\,\,small\,\,wrt.}\,\, L\}, \subseteq_o)$$
be a partially ordered set. 
If $L\neq\emptyset$ then $\mathcal{A}\neq\emptyset$, because each oriented sub-face $e$ of $E\in L$
can be refined to small oriented sub-boxes wrt. $L$. Let $\mathcal{M}$ be the set of maximal elements in $\mathcal{A}$. 
These are exactly the elements that are an intersection of a face $\partial$ of some $E\in L$ with a maximal number of boxes 
in $L$ s.t. the intersection is still $(d-1)$-dimensional.  
It follows that $\mathcal{M}$ is finite. Moreover, each $e\in\mathcal{A}$ is an oriented sub-box of a unique element $e'$ in $\mathcal{M}$. We define the equivalence class $[e]$
of some $e\in\mathcal{A}$ to be the set $\{g\in\mathcal{A}\,|\,g\subseteq_o e'\}$ for the unique $e'\in\mathcal{M}$.
For $e\in\mathcal{A}$, let $S_e$ be the subset of $L$ containing all $(E,i)\in L$ such that $e\hr E$.
If $e\subseteq_o e'\in\mathcal{M}$, then $S_e=S_{e'}$, so we may define the set $S_{[e]}$ for the equivalence class $[e]$. 
The balance property says that for each $e\in\mathcal{A}$, we have $|S_{[e]}|=|S_{[-e]}|$.
For each $e\in\mathcal{M}$, define the bijection $P_{[e]}: S_{[e]}\to S_{[-e]}$ in such a way that 
$P_{[-e]}=P_{[e]}^{-1}$ for all $e\in\mathcal{A}$.

Let $l\in \{1,\ldots,d+1\}$ and $s\in\{+,-\}$, $L^{sel}$ be the subset of $L$ containing all $(E,i)$
s.t. $(SV_E)_l=s$, let $L-L^{sel}$ be the set of non-selected boxes 
and assume that
$\mathit{faces}=\sum_{E\in\,L^{sel}}\,\mathit{faces}(E)$ was created in the Deg function body. 
We will represent $\mathit{faces}$ as a set of elements $(e,(E,i))$ such that $e\in \mathit{faces}((E,i))$ was created
as an oriented sub-face of $(E,i)\in L^{sel}$ in the Deg function body. Note that
for a particular $(E,i)\in L$, $e\hr E$ cannot be contained more than once in the list $\mathit{faces}((E,i))$, so 
$\mathit{faces}(E,i)$ contains each of its element exactly once, and hence each $(e, (E,i))$ represents a unique element of the $\mathit{faces}$ list. 
In~this set representation of $\mathit{faces}$, we ignore the order of the list. Note that the balancedness of $\mathit{faces}$, that we
want to prove, does not depend on the order of $\mathit{faces}$.

Let $(e,(E,i))\in \mathit{faces}$ and $x\hr e$ be so that $e$ is small wrt. $\mathit{faces}$ (this means that for each $(g,(E,i))\in \mathit{faces}$, 
either $x\subseteq g$ or $x\cap g$ is at most $(d-3)$-dimensional).
Let $T_x$ be the subset of $\mathit{faces}$ containing all $(g,(E,i))\in\,\mathit{faces}$ s.t. $x\hr g$. We want to show that $|T_x|=|T_{-x}|$. 
Let $x'\subseteq_o x$ be so small that for each $(E,i)\in L$, either $x\subseteq E$ or $x\cap E$ is at most $(d-3)$-dimensional 
(such a sub-box exists, because
it may be constructed as an intersection of $x$ with a finite number of boxes from $L$). 
$T_x=T_{y}$ holds for any oriented sub-box $y$ of $x$, so it is sufficient to show $|T_{x'}|=|T_{-x'}|$.
To prove this, we will construct a bijection $R_x:\,T_{x'}\to T_{-x'}$.

Let $(e_0, (E_0,i_0)) \in T_x$. This means that $e_0\in \mathit{faces}((E_0,i_0))$ for some $(E_0,i_0)\in L^{sel}$ and $x\hr e_0$.
Let $e_1$ be another sub-face of $E_0$ s.t. $x'\subseteq e_0\cap e_1$ and $e_1$ is small wrt. $L$ 
(such $e_1$ exists because of the condition that $x'$ is small wrt. $L$).
The sub-faces $e_0$ and $e_1$ of $E_0$ are oriented compatibly, so $(-x')\hr e_1$ and $e_1\in\mathcal{A}$. $E_0$ has up to equivalence only two sub-faces $e_0, e_1\in\mathcal{A}$ 
containing $x'$ so $[e_1]$ is uniquely determined by $x$ and $(e_0, (E_0, i_0))$.
If there exists some $(F,j)\in L^{non}$ s.t. $e_1\subseteq F$, then $e_1\subseteq_o e_1'\hr E_0$ for some $e_1'$ such that 
$(e_1', (E_0,i_0))\in \mathit{faces}$ and $(e_1', (E_0, i_0))\in T_{-x}$.
In that case, we define $R_x((e_0, (E_0, i_0))):=(e_1', (E_0, i_0))$. 
Otherwise, $e_1$ is not a subset of any non-selected box, and for each $(F,i)\in L$, $e_1\hr F$ implies $(F,i)\in L^{sel}$. 
Let $(E_1, i_1):=P_{[e_1]}((E_0, i_0))$. 
We know that $(E_1, i_1)\in L^{sel}$ and $x'\hr -e_1\hr E_1$. We again find a box $e_2$ in $E_1$ such that the intersection $e_1 \cap e_2$ contains $x'$ and $(-e_1)$ and $e_2$
are oriented compatibly,  so $-x'\hr e_2\hr E_1$.
In this way, we construct a sequence of boxes $e_j$ and elements $(E_j, i_j)$ such that $-x'\hr e_{j+1}\hr E_j$ for $j\geq 0$, 
$-e_j\hr E_j$ for $j\geq 1$, all $e_j$ are small wrt. $L$, $P_{[e_j]}((E_{j-1}, i_{j-1}))=(E_j, i_j)$ and
$e_j$ is not a subset of any non-selected box for $j=1,\ldots, p$. If $e_{p+1}$ is a subset of some non-selected box,
then $e_p\subseteq e_{p+1}'\in \mathit{faces}((E_p, i_p))$ and we define $R_x((e_0,(E_0, i_0)):=(e_{p+1}', (E_p, i_p))$.

It remains to prove that $R_x$ is correctly defined, i.e. that for some finite $p\in\N$, $e_{p+1}$ will be a subset of some non-selected box, 
and that $R_x$ is a bijection.
First we show that this procedure terminates. Assume, for contradiction, that the sequence $\{[e_j],(E_j, i_j)\}_j$ is infinite. 
Because there exists only a finite number of $(E_j,i_j)\in L$ and only a finite number of $[e_j]$, the sequence is periodic.
Let $k$ be the minimal index such that $([e_k], (E_k, i_k))=([e_{l}], (E_l, i_l))$ for some $l<k$. 
If $l>0$, then $(E_l,e_l)=P_{[e_{l}]}((E_{l-1}, i_{l-1}))$ and $(E_{l-1},i_{l-1})=P_{[-e_{l}]}((E_{l}, i_{l}))$ due to the assumption
$P_{[e_l]}=P_{[-e_l]}^{-1}$. It follows that $E_{l-1}$ is uniquely determined by $([e_l], (E_l,i_l))$ and 
$(E_{l-1},i_{l-1})=(E_{k-1},i_{k-1})$.  From the construction, we know that $-x'\hr e_{l}\hr E_{l-1}$. However,
in $E_{l-1}$, there exists up to equivalence a unique $-e_{l-1}\hr E_{l-1}$ s.t. $x'\hr (-e_{l-1})\hr E_{l-1}$. 
So, we proved that $([e_{l-1}], (E_{l-1}, i_{l-1}))=([e_{k-1}], (E_{k-1}, i_{k-1}))$, contradicting the assumption that
$k$ was the minimal index with such equality. If $l=0$ and $[e_k]=[e_0]$,
then the fact that $e_0$ is a subset of some non-selected box contradicts the assumption
that for each $i>0$, $e_i$ is not a subset of any non-selected box.

Finally, note that if $R_x(e_0, (E_0, i_0))=(e_{p+1},(E_{p}, i_p)$, then $R_{-x}(e_{p+1},(E_{p}, i_p)=(e_0, (E_0, i_0))$,
because each $([e_j], (E_j,i_j))$ is uniquely determined by $[e_{j+1}]$ and $(E_{j+1}, i_{j+1})$. 
So, starting with $(e_{p+1}, (E_p, i_p))$ will just reverse the order and we will eventually come to some $\tilde{e_0}$ s.t.
$\tilde{e_0}$ is a sub-face of $(E_0, i_0)$, $\tilde{e_0}$ is small wrt. $\mathit{faces}$ and $\tilde{e_0}$ is a subset of some non-selected box. It follows that 
$\tilde{e_0}$ is an oriented sub-box of the unique $(e_0,(E_0, i_0))\in \mathit{faces}$. This proves that $R_{-x}=R_x^{-1}$ and that $R$ is a bijection.
$\square$
%
% \\
%\bigskip \\
%{\bf Acknowledgements}  \\
%\\
%This work was supported by M{\v S}MT grant number OC10048 and the Czech Science Foundation
%(GA{\v C}R) grant number P202/12/J060 with institutional support RVO:67985807.

\bibliographystyle{abbrv}
\bibliography{sratscha}

\end{document}